\documentclass[onecolumn,11pt]{article}  

\pdfoutput=1  
 
\usepackage{titlesec}

\usepackage[skip=2pt,font=scriptsize]{caption}

\usepackage{flushend}

\usepackage{cite,url}

\usepackage[margin=0.97in,top=.99in,bottom=1.1in, textheight=22.51cm,textwidth=16.51cm, heightrounded]{geometry}
 
\usepackage[colorlinks,bookmarksopen,bookmarksnumbered,citecolor=red,urlcolor=red]{hyperref}

\usepackage[usenames,dvipsnames]{color}
\usepackage{graphicx,wrapfig,subfigure}

\usepackage[cmex10]{amsmath}
\usepackage{amsfonts,amssymb,mathrsfs}   
\usepackage{amsthm}
\usepackage[subnum]{cases}

\usepackage{etoolbox}
\AtBeginEnvironment{align}{\setcounter{subeqn}{0}}
\newcounter{subeqn} 
\renewcommand{\thesubeqn}{\theequation\alph{subeqn}}%
\newcommand{\subeqn}{%
  \refstepcounter{subeqn}
  \tag{\thesubeqn}
}

\newlength{\noteWidth}
\long\def\notes#1{\ifinner
          {\footnotesize #1}
          \else
          \marginpar{\parbox[t]{\noteWidth}{\raggedright\footnotesize #1}}
      \fi\typeout{#1}}

\def\notes#1{}

\def\spm#1{\notes{SPM:  #1}}

\def\Prob{{\sf P}} 
\def\Expect{{\sf E}}

\def\epsy{\varepsilon}

\def\varble{\,\cdot\,}

\newcommand{\field}[1]{\mathbb{#1}}

\def\Re{\field{R}}

\def\Co{\field{C}}

\def\One{\mbox{\rm{\large{1}}}}
\def\Zero{\mbox{\rm{\large{0}}}}

\def\transpose{{\hbox{\it\tiny T}}}

\newcounter{rmnum}
\newenvironment{romannum}{\begin{list}{{\upshape (\roman{rmnum})}}{\usecounter{rmnum}
\setlength{\leftmargin}{22pt}
\setlength{\rightmargin}{12pt}
\setlength{\itemsep}{2pt}
\setlength{\itemindent}{1pt}
}}{\end{list}}

\def\barq{\bar q}

\def\tilA{\widetilde{A}}
\def\tilS{\widetilde{S}}

\def\state{{\sf X}}

\def\clB{{\cal B}}

\def\clE{{\cal E}}
\def\clF{{\cal F}}
\def\clG{{\cal G}}

\def\clM{{\cal M}}

\def\clW{{\cal W}}
\def\clX{{\cal X}}

\def\bfmath#1{{\mathchoice{\mbox{\boldmath$#1$}}%
{\mbox{\boldmath$#1$}}%
{\mbox{\boldmath$\scriptstyle#1$}}%
{\mbox{\boldmath$\scriptscriptstyle#1$}}}}

 \def\bfmA{\bfmath{A}}
\def\bfmB{\bfmath{B}}

\def\bfmD{\bfmath{D}}

\def\bfmG{\bfmath{G}}

\def\bfmN{\bfmath{N}}  
\def\bfmQ{\bfmath{Q}}

\def\bfmS{\bfmath{S}}

\def\bfmW{\bfmath{W}}

\def\bfmX{\bfmath{X}}

\def\bfPsi{\bfmath{\Psi}}

\def\bfzeta{\bfmath{\zeta}}

\def\bfDelta{\bfmath{\Delta}}

\def\Cov{\text{Cov}\,}   
\def\diag{\,\text{\rm diag} }

\graphicspath{{figures/}}

\def\Ebox#1#2{%
\begin{center}
\includegraphics[width= #1\hsize]{#2} 
\end{center}}

\def\eqdef{\mathbin{:=}}

 \def\FRAC#1#2#3{\genfrac{}{}{}{#1}{#2}{#3}}

\def\half{{\mathchoice{\FRAC{1}{1}{2}}%
{\FRAC{1}{1}{2}}%
{\FRAC{3}{1}{2}}%
{\FRAC{3}{1}{2}}}}

\def\tilO{\widetilde O}

\def\Fig#1{Fig.~\ref{#1}}

\newtheorem{theorem}{Theorem}[section]

\newtheorem{proposition}[theorem]{Proposition}
\newtheorem{lemma}[theorem]{Lemma}

\def\Lemma#1{Lemma~\ref{#1}}
\def\Proposition#1{Proposition~\ref{#1}}
\def\Theorem#1{Theorem~\ref{#1}}

\def\Section#1{Section~\ref{#1}}

\def\Appendix#1{Appendix~\ref{#1}}

\def\Rtot{R^{\text{\tiny tot}}}
\def\haRtot{\widehat{R}^{\text{\tiny tot}}}

\def\nudist{\gamma}

\def\Sigmatot{\Sigma^{\text{\tiny tot}}}

\def\tilpi{{\tilde \pi}}

\def\haf{{\hat f}}

\def\haI{{\hat I}}

\def\hapi{{\hat\pi}}

\def\haSigma{{\widehat{\Sigma}}}

\def\haSigmatot{\widehat{\Sigma}^{\text{\tiny tot}}}

\def\ind{\field{I}}

\def\Re{\field{R}}
\def\nat{\field{Z}_+}
\def\ZZ{\field{Z}}

\def\diagpiload{\Lambda^{\Gamma}}
\def\piload{\Gamma}
\def\bfpiload{\bfmath{\Gamma}}

\def\psd{\text{\rm S}}
\def\psdder{\text{\textrm{S}}^{\text{\tiny(2)}}}

\def\As#1{Assumption~\textbf{#1}}
\def\SF{H}  

\def\KL{\text{\rm D}}  
\def\haKL{\widehat{\text{\rm D}}}

\def\probA{\psi_a}
\def\probB{\psi_b}

\def\LpB{\Lambda_b}

\def\haR{\widehat R}

\def\hapsd{\widehat{\text{\rm S}}}

\def\Prop#1{Prop.~\ref{#1}}

\title{Ergodic Theory
for Controlled Markov Chains
\\
with Stationary Inputs}

\author{Yue Chen, 
Ana Bu\v{s}i\'c, and Sean Meyn
\thanks{This research is supported by the NSF grants CPS-0931416 and CPS-1259040, the French National Research Agency grant ANR-12-MONU-0019, and US-Israel BSF Grant 2011506.}
\thanks{
Y.C. and S.M. are with the Department of Electrical and Computer
Engg.\ at the University of Florida, Gainesville. A.B.\ is with Inria and the Computer Science Dept. of \'Ecole Normale Sup\'erieure, Paris, France.}%
}

\begin{document}

\maketitle
\begin{abstract} 
 Consider a stochastic process $\bfmX$ on a finite state space $ \state=\{1,\dots, d\}$. 
 It is conditionally Markov, given a real-valued  `input process' $\bfzeta$. This
is assumed to be small,  which is modeled through the scaling,
\[
\zeta_t = \epsy \zeta^1_t, \qquad 0\le \epsy \le 1\,,
\]
where $\bfzeta^1$ is a bounded stationary process.  The following conclusions are obtained, subject to smoothness assumptions on the controlled transition matrix and   a   mixing condition on $\bfzeta$:
\begin{romannum}
\item 
A stationary version of the process is constructed, that is coupled with a stationary version of the Markov chain $\bfmX^\bullet$ obtained with $\bfzeta\equiv 0$.  The triple $(\bfmX, \bfmX^\bullet,\bfzeta)$ is a jointly stationary process 
satisfying  
\[
\Prob\{X(t) \neq X^\bullet(t)\} = O(\epsy)
\]
Moreover, a second-order Taylor-series approximation is obtained:
\[
\Prob\{X(t) =i \}  =\Prob\{X^\bullet(t) =i \}  + \epsy^2 \varrho(i) + o(\epsy^2),\quad 1\le i\le d,
\]
with an explicit formula for the vector $\varrho\in\Re^d$.   
\item 
For any $m\ge 1$ and any function $f\colon \{1,\dots,d\}\times \Re\to\Re^m$,  the stationary stochastic
process $Y(t) = f(X(t),\zeta(t))$ has a power spectral density $\psd_f$ that admits   a second order Taylor series expansion:
A function $\psdder_f\colon [-\pi,\pi] \to \Co^{ m\times m}$ is constructed such that 
\[
\psd_f(\theta) = \psd^\bullet_f(\theta) + \epsy^2 \psdder_f(\theta) + o(\epsy^2),\quad \theta\in [-\pi,\pi] .
\]
An explicit formula for the function $\psdder_f$ is obtained, based in part on the bounds in (i).
\end{romannum}
The results are illustrated using a version of the  timing channel of  Anantharam and Verdu.


\end{abstract}

%
\clearpage

\section{Introduction}

This paper concerns second-order ergodic theory for a controlled Markov chain.  Consider for the sake of illustration a stochastic process $\bfmX$ on a finite state space $ \state=\{1,\dots, d\}$,  which evolves together with a real-valued stationary sequence $\bfzeta$ and an i.i.d.\ sequence $\bfmN$
according to the recursion,
\begin{equation}
X_{t+1} =\varphi(X_t,\zeta_t, N_{t+1}),\qquad t\in \ZZ
\label{e:nonlinSS}
\end{equation}
where $\varphi\colon\state\times\Re^2\to \state$ is Borel measurable.
In the special case $\bfzeta\equiv 0$ we denote the solution $\bfmX^\bullet$, which is a time-homogeneous Markov chain.

  Among the questions we might ask are,
\begin{romannum}
\item If we are given only the function $\varphi$,  the stationary process
$\bfzeta$, and the i.i.d.\ sequence $\bfmN$, does there exist a stationary solution to \eqref{e:nonlinSS}?

\item
Can we compute statistics of $\bfmX$, such as the marginal distribution?

\item  If the answer to (ii) is no, can we \textit{approximate} the statistics of $\bfmX$?   How do these statistics compare with the stationary Markov chain $\bfmX^\bullet$?
\end{romannum}
Under the assumptions imposed in this paper, we construct the joint stationary process  $(\bfmX ,\bfmX^\bullet,\bfzeta)$ on the same probability space.

An answer to (ii) requires special assumptions on the model that are far beyond the scope of this paper.  Instead we consider a setting in which $\bfzeta$ is small, and obtain approximations to address question (iii).  It is simplest to consider a family of processes, parameterized by a small constant $\epsy>0$,
\[
\zeta_t = \epsy \zeta^1_t,
\]
where $\bfzeta^1=\{\zeta^1_t\}$ is a bounded sequence.  The construction of  $(\bfmX ,\bfmX^\bullet,\bfzeta)$ is obtained so that for each $t$,
\begin{equation}
\Prob\{X_t\neq X_t^\bullet\} = O(\epsy)
\label{e:tilOX}
\end{equation}
This then provides tools to address (iii).

In order to apply techniques from second-order statistics, the process is lifted to the 
unit simplex in $\Re^d$.  We denote 
\begin{equation}
\piload_t = e^j,  \qquad \text{when $X_t=j$,}
\label{e:piloadX}
\end{equation}
where $e^j$ denotes the $j$th standard basis element in $\Re^d$.   
This is a standard construction;
it is useful since the evolution of $\{\piload_t \}$ can be expressed as a linear state space model driven by an uncorrelated ``noise process'' (see \eqref{e:piP} below).   
This linear representation is used in   \cite{lipkrirub84} to construct a Kalman filter for a time-homogeneous Markov chain (without the input $\bfzeta$), and these results are extended to a class of controlled Markov chains in \cite{chebusmey15}.

The initial motivation for  \cite{chebusmey15}, as well as the research described here, is application to distributed control for the purposes of ``demand dispatch'' using distributed resources in a power grid.      The results of the present paper are applied in   \cite{chebusmey15c} to obtain performance approximations in the same power grid model.  Similar bounds were previously obtained in \cite{chebusmey14}, but this is the first paper to obtain an exact second-order Taylor series approximation for second-order statistics.  
\spm{should we be more specific?  Our 2014 paper did not include complete proofs, and did not provide a Taylor series approximation because the second order term  $ V_t^\transpose \zeta_t^2 $ in $D_{t+1}$ 
was neglected:  see \Prop{t:piload_linear}
}

\notes{
See this paper Raginsky suggested,  cite{berwu98}.
Just found this on flight to Banff: cite{linlem04}.  Useful?
}

 \notes{ 
Kalman filtering in triplet Markov chains  (IEEE Signal processing)
Generalized dynamic linear models for financial time series
 ... and ...
 Diffusion Approximation for Bayesian Markov Chains}

The main contribution of this paper is to obtain tight approximations for the joint auto-correlation function for $(\bfpiload,\bfzeta)$, and hence also its power spectral density.   To obtain these results requires the coupling bound \eqref{e:tilOX}, 
a second order Taylor series expansion for $\pi^\epsy = \Expect[\piload_t]$ in steady-state,
and surprisingly complex calculations for a linearized model.  

We are not aware of other related results in the literature.  However,   in some of our approximations we borrow one technique from  \cite{sch68} -- the use of the fundamental matrix appears in the approximation of $\pi^\epsy$;  see \eqref{e:funMatrix}
for a definition, and further explanation following this equation.

The main results are summarized in \Section{s:main}, with all of the technical proofs contained in an appendix.  \Section{s:ex} contains numerical results for an application to information theory -- a variant of the timing channel introduced in  \cite{anaver06}.  Conclusions and directions for future research are contained in
\Section{s:conc}.

\section{Model and main results}
\label{s:main}

Consider an irreducible and aperiodic Markov chain $\bfmX^\bullet$ evolving on a finite state space $\state=\{1,\dots, d\}$, and transition matrix $P_0$.   This admits a stationary realization on the two sided time-interval $\ZZ$,  whose marginal distribution $\pi_0$ is the invariant probability mass function (pmf) for $P_0$, satisfying $\pi_0 P_0 = \pi_0$. 
The goal of this paper is to investigate how the statistics change when the dynamics are subject to an exogenous disturbance.  

\subsection{Controlled Markov model}

The stochastic process $\bfmX$  considered in this paper also evolves on the finite state space $ \state$.   The ``disturbance'' in the controlled model is a one-dimensional stationary process denoted  $\bfzeta=\{ \zeta_t : -\infty < t<\infty\}$.   A controlled transition matrix $\{P_\zeta : \zeta\in\Re\}$ describes the dynamics of the process:
\begin{equation}
\Prob\{ X_{t+1}=k \mid  \zeta_s,  \  X_s : s\le t\} = P_{\zeta}(j,k) \qquad \text{when $\zeta_t=\zeta$, $X_t=j$.}
\label{e:Xdynamics}
\end{equation}
It is assumed that $P_\zeta$ is a smooth function of $\zeta$, and that
  $P_0$ is the transition matrix for $\bfmX^\bullet$.

\begin{wrapfigure}{r}[-3pt]{.3\hsize}
\vspace{-.25em}
\Ebox{.9}{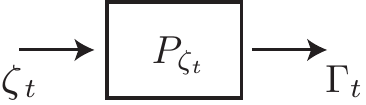}
\vspace{-.5em}
\caption{Controlled Markov model.}
\label{f:mdp}
\end{wrapfigure}

Since $\bfmX$ is no longer Markovian,  we cannot apply standard Markov chain theory to investigate properties of a stationary version of $\bfmX$.  Instead we apply linear systems theory, and for this we require a linear systems representation for the controlled stochastic process.  

This is obtained by embedding the process in $\Re^d$ through the indicator process
$\piload_t$ defined in \eqref{e:piloadX}.  Linear dynamics are obtained by considering a specific realization of the model.  We assume that there is a $d\times d$ matrix-valued function $\clG$ and an i.i.d.\ sequence $\bfmN$ for which,
\begin{equation}
	\piload_{t+1}=\piload_t G_{t+1} , \quad G_{t+1} = \clG(\zeta_t, N_{t+1})
\label{e:piG}
\end{equation}
It is assumed moreover that $\bfmN$ is independent of $\bfzeta$,  and that the entries of $\clG$ are zero or one, with
\[
\sum_{k=1}^d \clG_{j,k}(z,n) =1,  \qquad \text{for all $j, z, n$.}
\]
It follows from \eqref{e:Xdynamics} that for each $t$,
\begin{equation}
\label{e:ExpectG=P}
\Expect[\clG(\zeta_t, N_{t+1}) \mid \zeta_{-\infty}^\infty ] = P_{\zeta_t} \, .
\end{equation}

The random linear system \eqref{e:piG} illustrated in \Fig{f:mdp} is the focus of study in this paper.  The sequence $\bfpiload=\{ \piload_t \}$ is viewed as a state process, that is driven by the disturbance (or ``input'')  $\bfzeta$.  The state process evolves on the extreme points of the unit simplex in $\Re^d$.
We let $\bfpiload^\bullet =\{ \piload^\bullet_t \}$ denote the stationary Markov chain obtained with $\bfzeta \equiv 0$.



\notes{We wrote, It is assumed that $\bfpiload^\bullet $ is independent of $\bfzeta$.  This actually follows from our construction -- it isn't assumed}

Our main assumptions   are summarized in the following:
\begin{romannum}
\item[\textbf{A1:}]
The transition matrix $P_0$ is irreducible and aperiodic. 
 The matrix valued function $P_\zeta$ is twice continuously differentiable ($C^2$) in a neighborhood of $\zeta=0$,    and   the second derivative is Lipschitz continuous.

\item[\textbf{A2:}]
$\zeta_t = \epsy \zeta_t^1$ where $\bfzeta^1=\{\zeta^1_t : t \in \ZZ\}$ is a   real-valued stationary stochastic process with zero mean. The following additional assumptions are imposed:
\begin{romannum}
\item It is bounded,  $|\zeta^1_t|\le 1$ for all $t$ with probability one.  Hence      $\sigma_{\zeta^1}^2=\Expect[(\zeta^1_t)^2]\le 1$.
\item  Its auto-covariance is absolutely summable:
\[
\sum_{t=0}^\infty |R_{\zeta^1}(t)| <\infty
\]

\end{romannum}  
\end{romannum}
The power spectral density $\psd_{\zeta^1}$ exists and is continuous under
\As{A2}~(ii).  It also admits a spectral factor, denoted $\SF_{\zeta^1}$:
\begin{equation}
\psd_{\zeta^1}(\theta) = \SF_{\zeta^1}(e^{j\theta})\SF_{\zeta^1}(e^{-j\theta}),\quad -\pi\le \theta\le \pi.
\label{e:spectralFactor}
\end{equation}
See \cite{cai88} for background.

\As{A1} is used to obtain the approximation of \eqref{e:piG} by an LTI (linear time invariant) system.   The following intermediate step is used in   \cite{lipkrirub84} to obtain a Kalman filter for an uncontrolled Markov chain.
The proof of \Proposition{t:CovDecomp} is given in \Appendix{t:proof_CovDecomp}.

\begin{proposition}
\label{t:CovDecomp}
The random linear system \eqref{e:piG} can be represented as 
\begin{equation}
	\piload_{t+1}=\piload_t P_{\zeta_t} + \Delta_{t+1}, 
\label{e:piP}
\end{equation}
where $\Delta_{t+1} = \piload_t(G_{t+1}-P_{\zeta_t})$.  This is  a martingale difference sequence, 
with covariance matrix
\begin{equation}
\Sigma^\Delta = 
	 \Cov ( \piload_t[G_{t+1}-P_{\zeta_t}] )
	=
	\Expect\bigl[
			\diag( \piload_{t+1} )  
				- P_{\zeta_t}^\transpose \diagpiload_t P_{\zeta_t}
\bigr]
\label{e:smoothedSigmaDelta}
\end{equation}
where $\diagpiload_t$ is the diagonal matrix with diagonal
entries $\{\piload_t(x^i) :1\le i\le d\}$.
Moreover, 
\begin{equation}
R_{\Delta,\zeta}(t) = 0, \quad \text{for all } t.
\label{e:RDz0}
\end{equation}
That is, $\bfDelta$ and $\bfzeta$ are uncorrelated. 
\qed
\end{proposition}

We next apply the second-order Taylor series approximation:
\[
P_{\zeta_t} =  P_0 + \clE  \zeta_t + \frac{1}{2} \clW  \zeta_t^2 + O(\epsy^3)\,,
\]
where $\clE$ and $\clW$ denote the first and second  derivatives of $P_{\zeta}$, evaluated at $\zeta=0$:
\[
\left. \frac{d}{d\zeta} P_\zeta \right|_{\zeta=0} = \clE, \qquad \left. \frac{d^2}{d\zeta^2} P_\zeta \right|_{\zeta=0} = \clW.
\]
The $O(\epsy^3)$ bound holds under the Lipschitz condition for the second derivative of $P_\zeta$.

The recursion \eqref{e:piP} can be approximated as
\[
	\piload_{t+1}=\piload_t (P_0 + \clE  \zeta_t + \frac{1}{2} \clW  \zeta_t^2) + \Delta_{t+1} + O(\epsy^3).
\]
This is the LTI approximation:
\begin{proposition}
\label{t:piload_linear}
The system \eqref{e:piG} can be represented by,
\begin{equation}
	\piload_{t+1}=\piload_t P_0 + D_{t+1} + O(\epsy^3)\, ,
\label{e:piD}
\end{equation}
where, $ D_{t+1} = B_t^\transpose \zeta_t + V_t^\transpose \zeta_t^2 + \Delta_{t+1}$, with
\begin{equation}
B_t^\transpose = \piload_t \clE,\qquad V_t^\transpose = \half  \piload_t \clW\, .
\label{e:BV}
\end{equation}\qed
\end{proposition}

Applying the LTI approximation \eqref{e:piD}, an approximation for the auto-correlation of $(\bfpiload, \bfzeta)$ is obtained from an approximation for the pair process  $(\bfmD, \bfzeta)$. Since $\bfmD$ is taken as a row vector, we use the following notation for the auto-correlation of $( \bfmD, \bfzeta )$:
\begin{equation}
R(t) = \begin{bmatrix}  R_D(t) & R_{D,\zeta}(t)
\\ 
R_{D,\zeta}(-t)^\transpose & R_\zeta(t)
\end{bmatrix}
\label{e:corr}
\end{equation}
where $R_D(t) =\Expect[D(t)^\transpose D(0)]$, $R_{D,\zeta}(t) =\Expect[D(t)^\transpose \zeta(0)]$,
and the expectations are taken in steady-state.
 
The \textit{existence} of a steady-state solution is established in \Prop{t:piload-zeta}
 that follows.  

\subsection{Correlation formulae and approximations}

Under Assumptions  we obtain a coupling result, which plays an important role in the approximations that follow.
We write  $
\piload_t = \piload_t^\bullet + \tilO(\epsy)  $ if 
\[
\Expect[ \| \piload_t  - \piload_t^\bullet  \| ]= O(\epsy)\,,
\]
which implies that \eqref{e:tilOX} also holds.
We adopt similar notation for other random variables.  The following result is proven in  \Appendix{s:couple}:

\begin{proposition}
\label{t:piload-zeta}
Under Assumptions~\textbf{A1} and \textbf{A2}, there exists $\epsy_0>0$ such that the following holds for each $\epsy\in (0,\epsy_0]$:
the two process  $\bfpiload$ and $\bfpiload^\bullet$ can be constructed   so that  
$(\bfpiload ,\bfpiload^\bullet,\bfzeta)$ is jointly stationary on the two-sided time interval $\ZZ$,  
$\piload_t^\bullet$   is independent of $\bfzeta$,  and  moreover
\begin{eqnarray} 
\piload_t &=& \piload_t^\bullet + \tilO(\epsy)
\label{e:piloadpiload}
\\[.2cm]
\Expect[\piload_t  \zeta_t ] &=&O(\epsy^2)
\label{e:piloadzeta}
\end{eqnarray}
Consequently, for the stationary process, 
\begin{equation}
B_t = B_t^\bullet + \tilO(\epsy),  \quad V_t = V_t^\bullet + \tilO(\epsy) ,  \quad \Delta_t = \Delta_t^\bullet + \tilO(\epsy)
\label{e:BVDeltaBdds}
\end{equation}
\qed
\end{proposition}

The following strengthening of Assumption~\textbf{A2} is useful in computations:
\begin{romannum}
\item[\textbf{A3:}] 
The transfer function $\SF_{\zeta^1}$ in \eqref{e:spectralFactor} is rational, with distinct poles $\{\rho_1, \ldots, \rho_{n_z} \}$ satisfying $|\rho_i|<1$ for each $i$. 
\end{romannum}
Under  \textbf{A2} and \textbf{A3} the auto-covariance function for $\bfzeta$ 
can be expressed as a sum of geometrically decaying terms,  
\begin{equation}
R_{\zeta}(t) = \epsy^2 \sum_{k=1}^{n_z}a_k\rho_k^{|t|}\, ,
\label{e:R_zeta}
\end{equation}
where the $\{a_k\}$ can be determined from    $\SF_{\zeta^1}$.
Approximations for the auto-correlation functions $R_{D,\zeta}(t)$ and  $R_D(t)$ in \eqref{e:corr} are given in \Theorem{t:R_D}.  

As in the perturbation theory of \cite{sch68}, one component in these approximations is based on the \textit{fundamental matrix}, 
\begin{equation}
U_1 = [I -P_0 +\One \otimes \pi_0]^{-1} 
\label{e:funMatrix}
\end{equation}
where $\One\otimes\pi_0$ denotes the matrix whose rows are identical, and equal to $\pi_0$.   
Because the chain is irreducible and aperiodic, this can be expressed as a power series expansion,
\[
U_1 = I + \sum_{k=1}^\infty [P_0 -\One \otimes \pi_0]^k
\]
The summand can also be expressed $ [P_0 -\One \otimes \pi_0]^k=[P_0^k -\One \otimes \pi_0]  $,  $k\ge 1$.  Hence convergence of the sum follows from the mean ergodic theorem,
\begin{equation}
\lim_{k\to\infty}P_0^k  = \One \otimes \pi_0
\label{e:ergo}
\end{equation}
where the rate of convergence is geometric.

\begin{theorem}
\label{t:R_D}
Suppose that Assumptions \textbf{A1} and \textbf{A2} hold,  and consider the stationary process $(\bfpiload ,\bfpiload^\bullet,\bfzeta)$ constructed in 
\Proposition{t:piload-zeta}, with $\epsy\in (0,\epsy_0]$.  Then,   for each $t$,
\begin{align}
R_{D,\zeta}(t) &= B R_\zeta(t-1) + O(\epsy^3)
  \label{e:R_Dz}
  \\[.2cm] 
R_D(t)
& =  R_{B\zeta}(t)
\refstepcounter{equation}
\subeqn 
\label{e:R_DBz}
\\
&\quad + R_\Delta(t) 
\subeqn 
\label{e:R_DDelta} 
\\
&\quad +  R_{B\zeta,\Delta}(t - 1) + R_{B\zeta,\Delta}^\transpose( -t-1)
\subeqn 
\label{e:R_DBzDel}
\\
& \quad +  R_{V\zeta^2,\Delta}(t - 1) + R_{V\zeta^2,\Delta}^\transpose( -t-1) 
\subeqn 
\label{e:R_DV}
\\ 
& \quad+ O(\epsy^3)  
\nonumber
\end{align}
in which $B^\transpose  = \pi_0 \clE$ in   \eqref{e:R_Dz}, 
and each component shown on the right hand side of \eqref{e:R_DBz}--\eqref{e:R_DV}  is given below:

\bigskip
\noindent
 \text{\rm\bf (a)} \ 
The auto-correlation $R_{B\zeta}(t) =\Expect[B_t\zeta_tB_0^\transpose \zeta_0]$ in \eqref{e:R_DBz}
 admits the approximation,
 \begin{equation}
R_{B\zeta}(t) =  (P_0^t \clE)^\transpose \Pi_0 \clE R_\zeta(t) + O(\epsy^3), \quad t \ge0
\label{e:R_bz}
\end{equation}
where  $\Pi_0 = \diag (\pi_0)$.

\vspace{.5cm}
\noindent
 \text{\rm\bf (b)} \   The covariance for the martingale-difference sequence $\bfDelta$ is given by
$R_\Delta(t) =0$ for $t\neq 0$, and
\begin{equation}
\begin{aligned}
R_\Delta(0) = 
\Sigma^\Delta   = \Pi_\epsy & - P_0^\transpose \Pi_\epsy P_0 
\\
   &   - \bigl[ P_0^\transpose \diag (R_{\piload,\zeta}(0)) \clE
 		+  \clE^\transpose \diag (R_{\piload,\zeta}(0)) P_0 \bigr]
\\
  &    - \half  R_\zeta(0) \bigl[  P_0^\transpose \Pi_0 \clW
  	+ 2  {\clE}^\transpose  \Pi_0   \clE  +  {\clW}^\transpose \Pi_0 P_0  \bigr]
\end{aligned}
\label{e:R_Delta}
\end{equation}
where  $\Pi_\epsy = \diag (\pi_\epsy)$,  with $\pi_\epsy=\Expect[\piload_t]$,  and
\begin{equation}
R_{\piload,\zeta} (t) =\Expect[\piload(t)^\transpose \zeta(0)],\quad t\in\ZZ\, .
\label{e:Rpiloadzeta}
\end{equation}
When $\epsy=0$,  eq.~\eqref{e:R_Delta}
 becomes
\begin{equation}
 \Sigma^{\Delta^\bullet} = \Pi_0 - P_0^\transpose \Pi_0 P_0
\label{e:SigmaDeltaBullet}
\end{equation}

\vspace{.5cm}
\noindent
 \text{\rm\bf (c)} \   The cross-covariance $ R_{B\zeta, \Delta}(t) =  \Expect[B_t\zeta_t\Delta_0]$ admits the approximation,
\spm{\rd{To be carefully checked:}}
\begin{equation} 
R_{B\zeta, \Delta}(t) =
\begin{cases}
   0 &  t < 0 
   \\[.2cm] 
   \clE^\transpose R_{\Delta^2,\zeta}(0) \quad + O(\epsy^3)
     												&  t=0 
  \\[.2cm]
   \clE^\transpose A^tR_{\Delta^2,\zeta}(-t)
   + 
   \clE^\transpose \sum_{i=0} ^{t-1} A^{t-1-i} \clE^\transpose A^i R_\zeta (t-i) \Sigma^{\Delta ^\bullet} \quad+ O(\epsy^3) 	\qquad					& t \ge 1
 \end{cases} 
\label{e:R_bzD}
\end{equation}
where $A = P_0^\transpose$ and
\begin{equation}
R_{\Delta^2,\zeta}(t)  = \Expect[\Delta_t^\transpose \Delta_t \zeta_0]
\label{e:RD2zDef}
\end{equation}

\vspace{.5cm}
\noindent
 \text{\rm\bf (d)} \   The cross-covariance $ R_{V\zeta^2, \Delta}(t)=  \Expect[V_t \zeta_t^2\Delta_0]$ admits the approximation,
\notes{
\rd{Yue, please check my definitions for the $R_*$}.
\\
Note that I have removed $\epsy^2$ here to reduce clutter.  For example,
$\sigma_{\zeta}^2 $
insteady of
$\epsy^2\sigma_{\zeta^1}^2 $}
\begin{equation}
\label{e:R_VzD}
R_{V\zeta^2, \Delta}(t)=
 \begin{cases}
     0 & t<0 \\
    \frac{1}{2} \sigma_{\zeta}^2    (P_0^t \clW) ^\transpose \Sigma^{\Delta^\bullet}  +O(\epsy^3) & t\ge 0
 \end{cases}
\end{equation}
\qed 
\end{theorem}

The derivation of \Theorem{t:R_D} is given in \Appendix{t:proof_R_D}.

\Theorem{t:R_D} leaves out an approximation for $\pi_\epsy$ that is required in \eqref{e:R_Delta}.  
It also leaves out an approximation for $ R_{\piload,\zeta}(0)$ required in \eqref{e:R_Delta},
 and approximations for $\{R_{\Delta^2,\zeta}(t) : t\le 0\}$ in \eqref{e:R_bzD}. 
 These are obtained in the following:

\begin{proposition}
\label{t:pi_epsy+R_Gz} 
The following hold under Assumptions  \textbf{A1} and \textbf{A2}:
\begin{romannum}
\item  The steady state mean admits the approximation,
\begin{equation}
\pi_\epsy = \pi_0 +     \xi U_1  + O(\epsy^3) 
\label{e:pi-epsy}
\end{equation}
where  $U_1$ is the fundamental matrix \eqref{e:funMatrix}, and
\begin{equation}
\xi = \bigl( R_{\piload,\zeta}(0) \bigr)^\transpose \clE  +\half \sigma_{\zeta}^2 \pi_0 \clW .
\label{e:xiUnderA1A2}
\end{equation}

\item For $t\ge 0$ we have,
\begin{equation}
\begin{aligned}
R_{\Delta^2,\zeta}(-t) 
&= 
\diag  (R_{\piload,\zeta}(-t-1) P_0) 
\\
&\qquad - P_0^\transpose \diag (R_{\piload,\zeta}(-t-1))P_0 
+ R_\zeta(t+1) \Expect [ \clX_\bullet^{(1)} ] +O(\epsy^3)
\end{aligned}
\label{e:RD2z}
\end{equation} 
where
$\Expect [ \clX_\bullet^{(1)} ] \eqdef  \diag (\pi_0 \clE) -\left(P_0^\transpose \Pi_0 \clE  + [P_0^\transpose\Pi_0 \clE]^\transpose \right)$.

\item The correlation $R_{\piload,\zeta}$ is approximated as the infinite sum,
\begin{equation}
 R_{\piload,\zeta}(t) 
= \epsy^2  \sum_{i=1}^\infty  \bigl( B^\transpose P_0^{i-1}\bigr)^\transpose R_{\zeta^1}(t-i) + O(\epsy^3),\qquad t\in\ZZ\,,
\label{e:RGzGen}
\end{equation}
in which $\| B^\transpose P_0^{i-1} \| \to 0$ geometrically fast as $i\to\infty$.
\end{romannum}
\qed
\end{proposition}

The proof of \eqref{e:pi-epsy}
is given in
\Appendix{s:pi-epsy}, 
\eqref{e:RD2z} is given in \Appendix{s:e:RD2z}, 
 and \eqref{e:RGzGen} is established
in \Appendix{s:RGzGen}.  The geometric bound on the limit $\| B^\transpose P_0^{i-1} \| \to 0$ follows from \Lemma{t:FP}, and the formula $B^\transpose  = \pi_0 \clE$.

\begin{figure}[htb]
   \Ebox{.5}{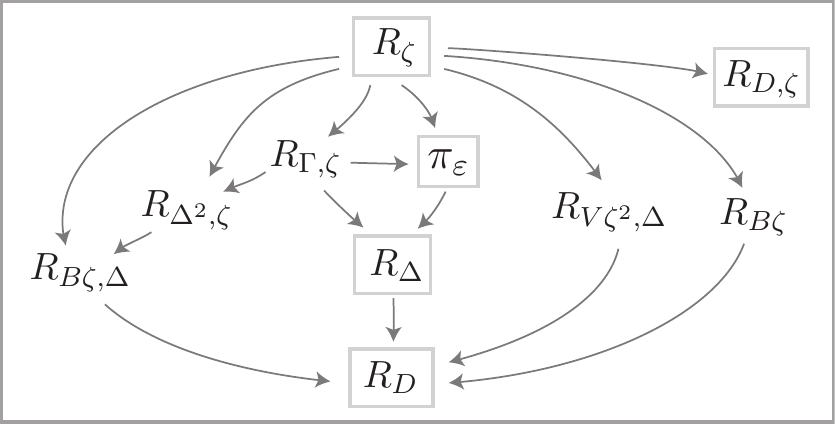}
   \caption{Dependency of autocorrelation functions involved in the approximations of $R(t)$ in \eqref{e:corr}.}
   \label{f:RiskSchematic}
 \end{figure}

The directed graph shown in \Fig{f:RiskSchematic} summarizes the dependency between all of these terms. For example,  the approximation of $R_{D,\zeta}$ only requires $R_\zeta$, and 
the covariance  $\Sigma^\Delta$ that defines $ R_\Delta$ is a function of $R_{\piload,\zeta}$ and $\pi_\epsy$.  The approximation of $R_D$ is a function of the four correlation functions shown  (as can also be seen from  \eqref{e:R_DBz}--\eqref{e:R_DV}).     The five boxed terms are those that are of interest to us directly;  the remaining five terms are introduced only to obtain a closed set of algebraic equations.

Closed-form expressions for the approximations in 
\Proposition{t:pi_epsy+R_Gz}  are possible under \textbf{A3}.   
The proof of
 \Proposition{t:R_Gz} is given  
in \Appendix{s:RGzGen}.

\begin{proposition}
\label{t:R_Gz}
Under  \textbf{A1}--\textbf{A3}, the row vector $\xi$ in \eqref{e:xiUnderA1A2} has the approximation, 
\begin{equation}
\xi =   \epsy^2  B^\transpose \sum_{k=1}^{n_z} a_k \rho_k[I-\rho_k P_0]^{-1}  \clE  +\half  \epsy^2 \sigma_{\zeta^1}^2 \pi_0 \clW + O(\epsy^3).
\label{e:xiUnderA3}
\end{equation} 
\qed
\end{proposition}

\subsection{Power spectral density approximations}

\Theorem{t:R_D} provides a second-order approximation of the auto-covariance function $\{R(t)\}$ defined in \eqref{e:corr}, which we denote $\{\haR(t)\}$.   In particular, $\haR_D(t)$ is defined as the sum of  \eqref{e:R_DBz}--\eqref{e:R_DV}.  Based on this and \Proposition{t:piload_linear} we obtain a second-order approximation $\{\haRtot(t)\}$  of the auto-covariance function $\{\Rtot(t)\}$ for the triple 
$(\bfpiload,\bfmD, \bfzeta)$.

The power spectral density (PSD) of a stationary process is the Fourier transform of its \textit{auto-covariance}.  This matrix-valued function is denoted
\[
\psd(\theta) = \sum_{t=-\infty}^{\infty} \Sigmatot(t) e^{-j\theta t} ,  \quad \theta\in\Re
\]
in which  $\Sigmatot(t) \eqdef \Rtot(t) - \mu  \mu^\transpose $, $ t\ge 0$,
with $\mu^\transpose = \Expect[(\piload_t,D_t,\zeta_t)]$. 

To define an approximation for $\psd$ we 
must obtain an approximation $\{\haSigma(t)\}$ that is summable.    It turns out that this is obtained from   $\{\haR(t)\}$ without normalization.  For each $t$,  the $(2d+1)\times(2d+1)$ matrix is decomposd as follows:
\[
\haSigmatot(t)  =\begin{bmatrix}  \haSigma_\piload(t) & \haSigma_{\piload, (D,\zeta)} (t)
 \\[.2cm]
						\haSigma_{(D,\zeta), \piload} (t)  &  \haSigma(t)
\end{bmatrix}
\]
in which $ \haSigma(t)  = \haR(t) $, and the remaining terms are what would be obtained by ignoring the $O(\epsy^3)$ error term appearing in \eqref{e:piD},  and replacing $P_0$ by its deviation
$P_0 -\One \otimes \pi_0$.   Denote $\tilA = (P_0 -\One \otimes \pi_0)^\transpose$, and
\[ 
 \haSigma_\piload(t) = \sum_{i,j=0}^\infty \tilA^i \haR_D(t-i+j) (\tilA^j)^\transpose 
\]
The matrix $ \haSigma_{\piload, (D,\zeta)} (t)$ is the $(d+1)$-dimensional column vector whose first $d$ components are 
\[
  \haSigma_{\piload,D}(t)  = \sum_{i=0}^\infty \tilA^i \haR_D(t-i)  
\]
and the final component is defined by the right hand side of \eqref{e:RGzGen}, ignoring the approximation error. This can be equivalently expressed,
\[
  \haSigma_{\piload,\zeta}(t)  = \epsy^2  \sum_{i=1}^\infty   \tilA^{i-1} B R_{\zeta^1}(t-i) 
\]
where we have used the fact that $ \tilA^k B= (P_0^\transpose)^k B$ since $\One^\transpose B=\sum B_i=0$.
\spm{need another lemma?}
Finally,  $\haSigma_{(D,\zeta), \piload} (t) = \haSigma_{\piload, (D,\zeta)} (-t)^\transpose$.
 
 Denote
\[
\hapsd(\theta) = \sum_{t=-\infty}^{\infty} \haSigmatot(t) e^{-j\theta t} ,  \quad \theta\in\Re
\]
It can be shown that the sequence $\{\haRtot(t)\}$ is absolutely summable, so that the approximation $\hapsd$ is a continuous bounded function of $\theta$.  The following is an immediate corollary to  \Theorem{t:R_D}:

\begin{proposition}
\label{t:psdDecompA} 
The approximation of the  power spectral density of the stationary sequence $\{ D_{t+1} =\Delta_{t+1}+
B_t^\transpose \zeta_t +  V_t^\transpose \zeta_t^2    \}$ can be expressed as the sum, 
\begin{equation}
\begin{aligned}
\hapsd_D(\theta) =   \hapsd_\Delta(\theta) +  \hapsd_{B\zeta}(\theta)  & + e^{-j\theta} \hapsd_{B\zeta,\Delta}(\theta)  +   \bigl[ e^{-j\theta} \hapsd_{B\zeta,\Delta}(\theta) \bigr]^* 
\\
& 	+ e^{-j\theta} \hapsd_{V\zeta^2,\Delta}(\theta)  +   \bigl[ e^{-j\theta} \hapsd_{V\zeta^2,\Delta}(\theta) \bigr]^* 
\end{aligned}
\label{e:psd_D}
\end{equation}
in which each approximation on the right hand side is obtained from is obtained through the Fourier transform of the corresponding approximations   \eqref{e:R_bz}--\eqref{e:R_VzD} in \Theorem{t:R_D}, and 
where   ``${}^*$'' denotes complex-conjugate transpose. 

The power spectral density approximation for $\bfpiload$ is given by,   
\[
\hapsd_\piload(\theta) =  [I e^{-j\theta} -\tilA]^{-1}  \hapsd_D(\theta)  [I e^{j\theta} -\tilA^\transpose]^{-1} 
\]
and the cross-power spectral density approximations are
\[
\begin{aligned}
\hapsd_{\piload,D}(\theta)  &=  [I e^{-j\theta} -\tilA]^{-1}  \hapsd_D(\theta)   
\\
\hapsd_{\piload,\zeta}(\theta)  &=     \epsy^2 
\sum_{t=-\infty}^{\infty}   \sum_{i=1}^\infty   \tilA^{i-1} B R_{\zeta^1}(t-i)  e^{-j\theta t}  
\end{aligned}
\]
\qed 
\end{proposition}

%
%
 \notes{Yue, you can keep proof in your thesis}

Under slightly stronger assumptions we obtain a uniform bound for this approximation.
The proof of \Proposition{t:psdDecomp} 
is given in \Section{s:psdDecomp}.

\begin{proposition}
\label{t:psdDecomp} 
Suppose that Assumptions~\textbf{A1} and \textbf{A2} hold, and in addition $R_\zeta(t)\to 0$ geometrically fast as $t\to\infty$.   Then, the uniform approximation holds:   For any $\varrho\in (0,1)$,
\[
\hapsd(\theta) =\psd(\theta)+O(\epsy^{2+\varrho}),  \quad \theta\in\Re
\] 
\qed
\end{proposition}

%

\def\jumpP{\gamma}

\section{Example:  bits through queues}
\label{s:ex}
 
The following example is motivated by the communication model of \cite{anaver06}.   There is a sender that wishes to send data to a receiver.  Neither has access to a communication channel in the usual sense.   Instead,   the sender manipulates  the timing of packets to a queue,   and the receiver gathers data through observations of the timing of departures from the queue.  

\subsection{Timing channel model}

To obtain a finite state-space model it is assumed that the queue size is bounded by $\barq $,  and arrivals are rejected if they cause an overflow.  The dynamics of the queue are described as a reflected random walk,
\begin{equation}
Q_{t+1} = \min\{\barq, \max (0,  Q_t - S_{t+1} + A_{t+1}) \}  ,\qquad t\ge 0
\label{e:CRW}
\end{equation}
In the nominal model in which $\bfzeta\equiv 0$, the pair process $(\bfmS,\bfmA)$ is i.i.d.\ on $\nat^2$. The sender wishes to manipulate the arrival process $\bfmA$, and the receiver observes the departure process $\bfmS$.  This manipulation is modeled through a scalar input sequence $\bfzeta$.

For simplicity, for the nominal model we restrict to the  M/M/1 queue:  The usual model evolves in continuous time, but after sampling using \textit{uniformization} one obtains \eqref{e:CRW}, in which $\bfmA$ a Bernoulli sequence,  and $S_t=1-A_t$ for each $t$. For each integer $n\in\state=\{0,1,\dots,\barq\}$, denote $n^+ = \min(n+1,\barq) $ and $n^- = \max(n-1,0)$.  If $0<\lambda<\half$ is the probably of success for $\bfmA$, we then have,
\[
\Prob\{Q(t+1) = n^+ \mid Q(t) = n\} = 1 - \Prob\{Q(t+1) = n^- \mid Q(t) = n\} = \lambda
\]
Its steady-state pmf is given by
\[
\pi_0^Q(n) = \kappa \rho^n
\]
where $\rho=\lambda/(1-\lambda)$, and $\kappa>0$ is a normalizing factor.

Recall that the receiver observes departures from the queue, which is equivalent to observations of the sequence $\bfmS$.  To estimate joint statistics we expand the state space to $X(t) = (Q(t), S(t))$, which evolves on the state space $\state = \{0,1,\dots,\barq\}\times\{0,1\}$.   The nominal transition matrix is defined as follows,
\[
\begin{aligned}
P_0((n,s),(n^+,0)) &= \lambda   
\\
 P_0((n,s),(n^-,1)) &= 1-\lambda  
\end{aligned}
\]
The first identity holds because a transition from $(n,s)$ to $(n^+,0)$ means that  $A_{t+1}=1$, in which case $S_{t+1}=1-A_{t+1}=0$.   The justification for the second identity is symmetrical.  The transition matrix is sparse: $P_0(x,x')=0$ for all but at most two values of $x'$, regardless of $x$.

The sender wishes to the manipulate timing of arrivals,  which motivates the following formulation for the controlled transition matrix:
\[
\begin{aligned}
P_\zeta((n,s),(n^+,0)) &=\lambda(1+\zeta) 
  \\
  P_\zeta((n,s),(n^-,1)) &=1-\lambda(1+\zeta) 
\end{aligned}
\]
in which $\zeta$ is constrained to the interval $[-1,1]$.
The state process evolves as the nonlinear state space model \eqref{e:nonlinSS}, with
\[
\begin{aligned}
Q_{t+1} &= \min\{\barq, \max (0,  Q_t -  1 +2\ind\{N_{t+1} \le  \lambda(1+\zeta_t) ) \} 
\\
S_{t+1} &=   \ind\{N_{t+1} >  \lambda(1+\zeta_t) \}  ,\qquad t\ge 0
\end{aligned}
\]
in which $\bfmN$ is i.i.d., with marginal equal to the  uniform distribution on $[0,1]$.

\begin{figure}[h]
\Ebox{.6}{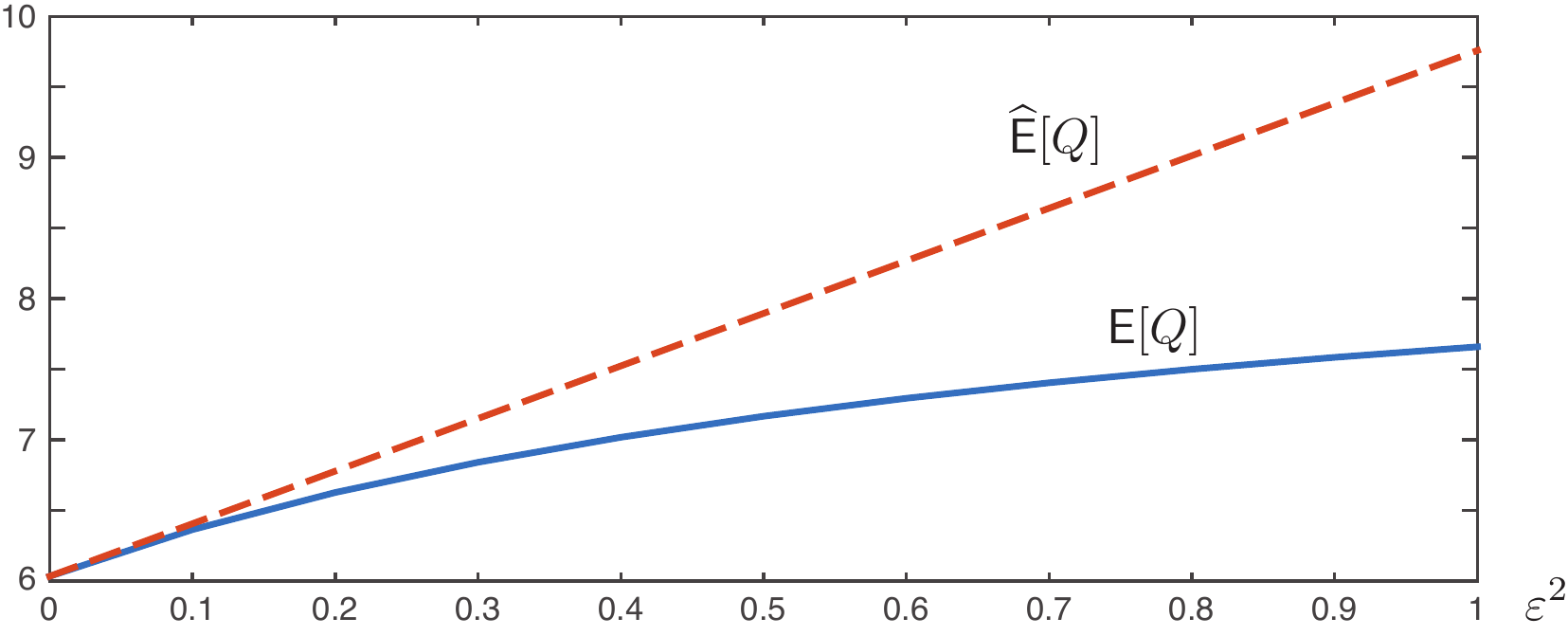} 
   \caption{Mean queue length grows approximately linearly in $\epsy^2$, for $\epsy^2\le 0.2$.  The error for $\epsy=1$ is approximately 25\%.}
   \label{f:meanQapprox}
\end{figure}

\Fig{f:meanQapprox} shows a comparison of the steady-state mean queue length as a function of $\epsy^2$ for a numerical example.  The linear approximation is obtained from the approximation of $\pi_\epsy$ given in  \Prop{t:pi_epsy+R_Gz}. Other statistics are shown in \Fig{f:PSDpiQapprox}.
Details can be found in \Section{s:Qnum}.

 \begin{figure}
\Ebox{1}{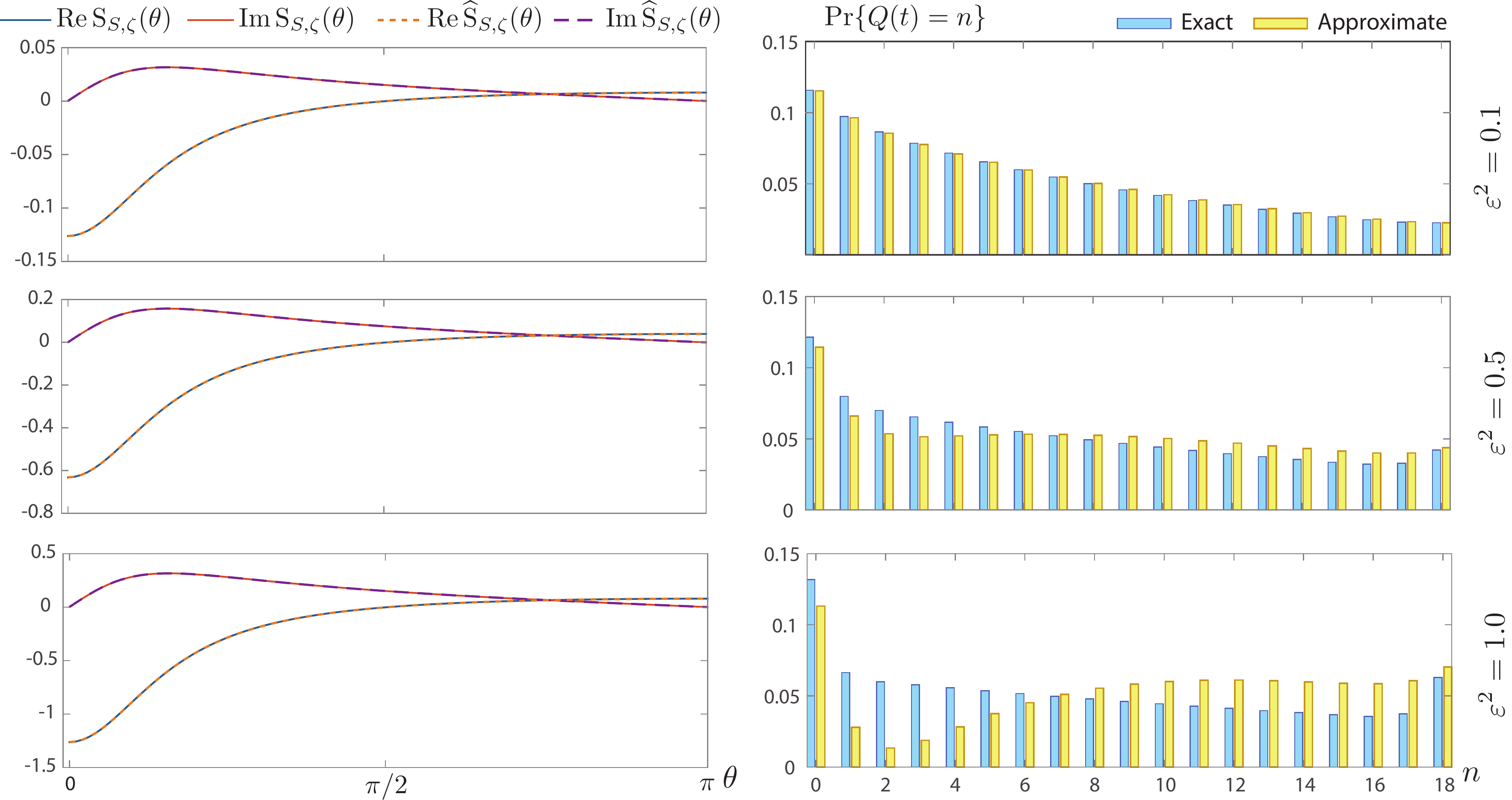}
   \caption{Experiments using $\gamma=0.4$.   The approximation for the cross power spectral density $\psd_{S,\zeta}$ appears to be exact for the entire range of $\epsy$. The approximation for the steady-state distribution of $\bfmQ$ is accurate for $\epsy^2\le 0.5$, but is very poor for $\epsy=1$.}
   \label{f:PSDpiQapprox}
 \end{figure}

\subsection{Second-order bound for mutual information}

The mutual information rate $I(\bfmS,\bfzeta)$ between $\bfmS$ and  $\bfzeta$ defines the capacity of this channel.  Letting $\chi^N$ denote the joint distribution of $\{ S_1,\dots, S_N,  \zeta_1,\dots,\zeta_N\}$,  and denoting the marginals   $\{ S_1,\dots, S_N\} \sim \chi^N_S$,  $\{  \zeta_1,\dots,\zeta_N\} \sim \chi^N_\zeta$, 
the mutual information rate is defined as the limit
\[
I(\bfmS,\bfzeta) = \lim_{N\to\infty}\frac{1}{N}   \KL(\chi^N \, \| \, \chi^N_S\times \chi^N_\zeta)
\]
where $\KL$ denotes relative entropy (i.e., K-L divergence) \cite{covtho91a}.   
In the following an approximation $\haKL$ is introduced,  and based on this an approximation to mutual information, 
\begin{equation}
\haI(\bfmS,\bfzeta) = \lim_{N\to\infty}\frac{1}{N}   \haKL(\chi^N \, \| \, \chi^N_S\times \chi^N_\zeta).
\label{e:haI}
\end{equation}

The approximation of relative entropy is given here in a general setting.  Let  $\probA$ and $\probB$ be probability measures on an abstract measurable space $(\clE,\clB)$.  For a measurable function $f\colon\clE\to\Re$ we let $\probA(f)$ denote the mean $\int f(x) \, \probA(dx)$.  The approximation is the non-negative functional defined as follows:
\begin{equation}
\haKL(\probA\|\probB) = \half \sup \Bigl\{  \frac{\probA(f)^2 } {\probB(f^2)} : \
					 \probB(f)=0 ,\  \probA(|f|)  <\infty,\quad \text{\it and}\quad  0<\probB(f^2)<\infty \Bigr\}
\label{e:haD}
\end{equation}
  
The proof of \Prop{t:haD}   
is contained in the Appendix.
\begin{proposition}
\label{t:haD}   
The following hold for any two probability measures   $\probA$ and $\probB$  on an abstract measurable space $(\clE,\clB)$.   Let $f^* = \log(d\probA/d\probB )$ denote  the log-likelihood ratio.
\begin{romannum}

\item The maximum in \eqref{e:haD}
 is given by  $\haf^* =  e^{f^*}-1 = d\probA/d\probB -1$,  whenever $\haKL(\probA\|\probB) $ is finite.
 
\item  There is a convex, increasing function $\kappa\colon\Re_+\to\Re_+$ that vanishes only at the origin, and such that the following bound holds for any two probability measures with bounded log-likelihood ratio:
\[
|\haKL(\probA\|\probB) - \KL(\probA\|\probB) |\le \kappa (\| f^* \|^3_\infty)
\]
where   $ \| f^* \|_\infty$ denotes the supremum norm.
\end{romannum}
\qed
\end{proposition}

Returning to the stochastic process setting, in the context of \eqref{e:haI}, we have for fixed $N$ the following correspondences:
\[
\probA = \chi^N,\qquad  \probB= \chi^N_S\times \chi^N_\zeta
\]
Consider for   $0\le n< N$ the function 
\[
f(S_1,\dots, S_N,  \zeta_1,\dots,\zeta_N) = \sum_{k=1}^{N-n}  \tilS_{k+n}\zeta_k
\]
in which $\tilS_t=S_t-\Expect[S_t]$. 
This has mean zero under $\probB$, and its mean under $\probA$ is,
\[
\probA(f) = 
\Expect_{\chi^N}[ f(S_1,\dots, S_N,  \zeta_1,\dots,\zeta_N)  ] = (N-n) \Sigma_{S,\zeta}(n)
\]
The second moment is also expressed in terms of autocorrelation functions:
\[
\begin{aligned}
\probB(f^2) &= 
\Expect_{\chi^N_S\times \chi^N_\zeta}[ f^2(S_1,\dots, S_N,  \zeta_1,\dots,\zeta_N)  ]  
\\
  &=\sum_{k=1}^{N-n} \sum_{\ell=1}^{N-n} \Expect_{\chi^N_S\times \chi^N_\zeta}
  		 \Expect [  \tilS_{k+n} \tilS_{\ell+n}\zeta_k \zeta_\ell]\\
  &=\sum_{k=1}^{N-n} \sum_{\ell=1}^{N-n}  \Sigma_S(k-\ell) \Sigma_\zeta(k-\ell) 
\end{aligned}
\]
For fixed $n$ and $N\gg n$ this admits the approximation
$
\probB(f^2) \approx (N-n) \psd_{S\times\zeta}(0)$,
where
\[
 \psd_{S\times\zeta}(0)=\sum_{m=-\infty}^\infty \Sigma_S(m)  \Sigma_\zeta(m) \, .
\]
This gives the large-$N$ approximation,
\[
\haKL(\probA\|\probB) \ge  \half    \frac{\probA(f)^2 } {\probB(f^2)} =   \half  \frac{\Sigma_{S,\zeta}(n)^2}{\psd_{S\times\zeta}(0)} N+O(1)
\]
While the derivation was performed for $n\ge 0$,  similar arguments establish the same bound
 for any integer $n$.
The approximation for mutual information rate is thus lower bounded,
\begin{equation}
\haI(\bfmS,\bfzeta) \ge  \half \sup_{-\infty<n<\infty} \frac{\Sigma_{S,\zeta}(n)^2}{\psd_{S\times\zeta}(0)}
\label{e:haI_secondOrder}
\end{equation}

This function class is of course highly restrictive.  A larger class of functions can be obtained by defining for each $n$ and each $\alpha,\beta \in\Re^{n+1}$,
 \[
f(S_1,\dots, S_N,  \zeta_1,\dots,\zeta_N) = \sum_{k=1}^{N-n}  S^\alpha_k\zeta_k^\beta,\qquad \textit{\it with}\quad S^\alpha_k= \sum_{m=0}^n \alpha_m \tilS_{k+m},
\
\
\zeta^\beta_k= \sum_{m=0}^n \beta_m \zeta_{k+m}.
\]
Formulae for $\probA(f)$ and $\probB(f^2)$ can be obtained as in the foregoing, yielding 
\[
\haI(\bfmS,\bfzeta) \ge \half   \frac{\Sigma_{S^\alpha,\zeta^\beta}(0)^2}{\psd_{S^\alpha\times\zeta^\beta}(0)}
\]

\subsection{Numerical experiments}
\label{s:Qnum}
  
%
%

In all of the numerical examples described here, $\lambda$ is chosen so that  $\rho=\lambda/(1-\lambda) = 0.9$, and  the upper bound appearing in \eqref{e:CRW}
 is  $\barq=18$.
 
A Markovian model was chosen for $\zeta^1$ so that exact computations can be obtained for the larger Markov chain.   A simple model was chosen, in which   $\bfzeta^1$ evolves on the three states $\{-1,0,1\}$.  The larger Markov chain $\Psi_t = (Q_t ,  S_t,\zeta^1_t)$ then evolves on a state space of size $6(1+\barq)$.

\begin{figure}
\Ebox{1}{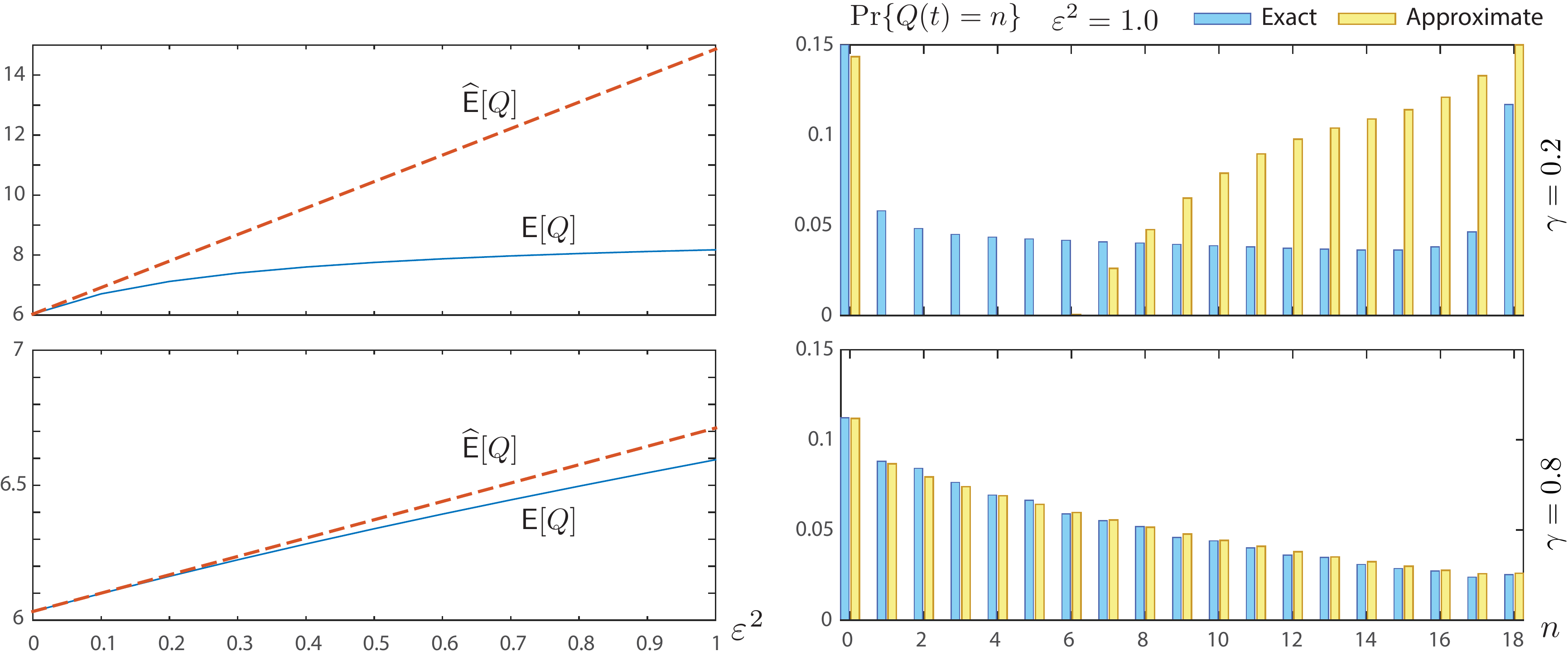}
   \caption{The top row shows results from numerical experiments with $\gamma=0.2$,  and the bottom row $\gamma=0.8$.  The approximation for $\pi^Q_\epsy(n)$ is nearly perfect in the latter case, even with $\epsy=1$.  With $\gamma=0.2$,  the second order approximation is poor for $\epsy\ge0.5$.}\label{f:gamma0208epsy1Q}
 \end{figure}

The three states are labeled $\{z^i: i=1,2,3\} = \{-1,0,1\}$.  For a fixed parameter  $\jumpP\in (0,\half)$, the transition matrix $K	$ is defined as follows.   First,  $\Prob\{\zeta_{t+1} = z^j \mid \zeta_t=z^i\} = \jumpP$ whenever $|z^j-z^i|=1$:
\[ 
K_{1,2} = K_{2,1} = K_{2,3} = K_{3,2}     = \jumpP 
\]
The remaining transition probabilities are $K_{1,1}=K_{3,3}=  1-\jumpP$,  and  $K_{2,2} = 1- 2\jumpP$.   The steady-state pmf $\mu_0$ is uniform, 
so the steady-state variance is
\[
\sigma^2_\zeta = ((-1)^2+0^2+1^2)/3 = 2/3
\]
Its autocorrelation is equal to its autocovariance:
 $R_\zeta(m) =\sigma^2_\zeta  (1-\jumpP)^{|m|}$.  
 
Unless explicitly stated otherwise, the results that follow use $\jumpP=0.4$,  so that the asymptotic variance  (the variance appearing in the Central Limit Theorem) is  
\[
\text{asym. variance} = \psd_\zeta(0) =
\sum_{k=-\infty}^\infty R_\zeta(k) = \Bigl( \frac{2}{\jumpP}  -  1\Bigr)\sigma^2_\zeta =4\sigma^2_\zeta
\]

Let $\hapi_\epsy = \pi_0 +   \xi U_1 $ (with $\xi$ and $U_1$ defined in  \eqref{t:pi_epsy+R_Gz}).   
The approximate pmf illustrated in the plots on the right hand side of \Fig{f:PSDpiQapprox} are defined by the first marginal,  $\hapi_\epsy^Q(n) =
\sum_{s=0,1}  \hapi_\epsy(n,s) $,  for $n=0,\dots, 18$.   The approximate steady-state queue length plotted in   \Fig{f:meanQapprox} is defined by
\[
\widehat{\Expect}[ Q] =   \sum_{n=0}^{18}   n\hapi_\epsy^Q(n)
\]
The steady state pmf for 
$\bfPsi$ was computed to obtain the exact steady-state mean $\Expect[Q]$,  which is the concave plot shown in  \Fig{f:meanQapprox}.  The approximations are accurate for $\epsy\le 0.7$.

   The approximations for the cross power spectral density shown on the left hand side of  \Fig{f:PSDpiQapprox} are remarkable.
   
   \spm{2016 deleted: -- it is not known why the approximation appears to be exact for all values of $\epsy$ and $\gamma$.
   }
    
The statistics of $\bfmQ$ and its approximations are highly sensitive to the parameter $\gamma$.   For $\gamma=0.8$,   the approximation of the steady-state mean $\Expect[Q]$ and approximations of $\pi^Q_\epsy(n)$ are nearly exact for the entire range of $\epsy$.  For $0<\gamma\le    0.2$ the approximations are accurate only for a very narrow range of $\epsy\sim 0$.    
  Results for $\gamma=0.2$ and $\gamma=0.8$ are shown in \Fig{f:gamma0208epsy1Q}.

  \begin{figure}
  \Ebox{.65}{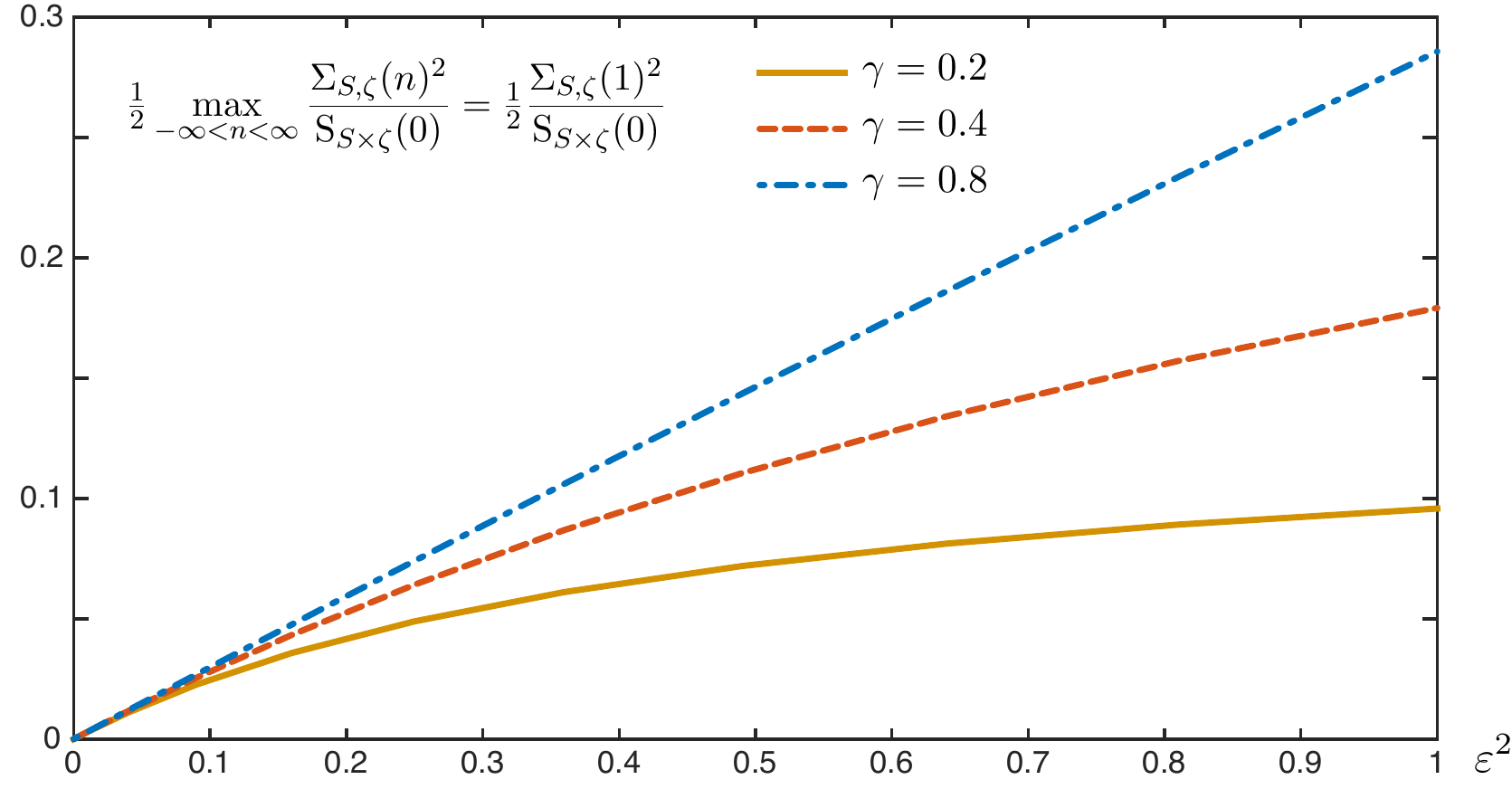} 
   \caption{Lower bound for $\haI(\bfmS,\bfzeta) $ as a function of $\epsy^2$ for three values of $\gamma$. }
   \label{f:mutualinfoApprox}
 \end{figure}

The approximation for mutual information in \eqref{e:haI_secondOrder} is defined as a maximum,   which was achieved at $n=1$ in each experiment:
\[  
\haI(\bfmS,\bfzeta) \ge 
\half   \max_{-\infty<n<\infty}
\frac{\Sigma_{S,\zeta}(n)^2}{\psd_{S\times\zeta}(0)}
=
\half    
\frac{\Sigma_{S,\zeta}(1)^2}{\psd_{S\times\zeta}(0)}
\]
Plots for four different values of $\gamma$ are shown in \Fig{f:mutualinfoApprox}.

The plots in  \Fig{f:mutualinfoApprox} use the approximations obtained in \Section{s:main}.  However,  plots obtained using the exact values of $\Sigma_{S,\zeta}(1) $  and $\psd_{S\times\zeta}(0)$ are indistinguishable. 

\spm{2016 ??
\\
The error in the approximation of $ \psd_{S\times\zeta}(0)$ obtained from \Prop{t:psdDecompA} is no more than 2\%\ in all of the experiments conducted.  
}

Unfortunately,  we cannot compare $\haI(\bfmS,\bfzeta) $ with the true mutual information rate.    This is a significant computational challenge that is beyond the scope of this paper.

\def\PPsi{P^{\tiny \Psi}}
\def\PO{\text{\tt P0}}
\def\Pder{\text{{\tt P\_der1}}}
%
%
%
%
%
%

\section{Conclusions}
\label{s:conc}

It is very surprising to obtain an exact second order Taylor series expansion for these second order statistics with minimal assumptions on the controlled Markov model.  The accuracy of the approximations obtained in numerical examples is also fortunate. The companion paper \cite{chebusmey15c} shows that these approximations are also accurate in applications to distributed control.
\spm{don't reference chebusmey14 here since this paper has a hack approximation -- not a second order Taylor series}
 Further work is needed to see if this will lead to useful bounds in applications to information theory.   
 \spm{2016  ok?}

\bibliographystyle{IEEEtran}
\bibliography{strings,markov,q}


\vspace{1cm}

\appendix

\centerline{\Large\bf Appendices}

\section{Coupling}
\label{s:couple}
 
We present here the proof of \Proposition{t:piload-zeta}.

We first obtain a recursion for the joint process $\bfPsi\eqdef (\bfpiload ,\bfpiload^\bullet)$, driven by $\bfzeta$, and two  i.i.d.\ sequences $\bfmN^\circ$ and $\bfmN^\bullet$,   each with marginal distribution uniform on $[0,1]$.  The three sequences $\bfzeta,\bfmN^\circ ,\bfmN^\bullet$ are mutually independent.

 Letting $\bfmW$ denote the 3-dimensional stationary stochastic process $(\bfzeta,\bfmN^\circ ,\bfmN^\bullet)$,  we   construct a function $F$ for which,
\begin{equation}
\Psi_{t+1} = F(\Psi_t, W_{t+1})
\label{e:piloadJoint}
\end{equation}
Since $\bfPsi$ evolves on a finite set,  the existence of a stationary solution follows from  
\cite{foskon04} (see Theorem~5 and the discussion that follows).

To construct the function $F$ it is enough to define the matrix sequence $\bfmG$ that appears in \eqref{e:piG}, and also the sequence $\bfmG^\bullet$ that defines the dynamics of $\bfpiload^\bullet$.  Each are based on the following definition:  for $\zeta\in\Re$ and $s\in[0,1]$, denote
\[
\clG_{i,j} (\zeta, s)  = \ind\Bigl\{ \sum_{k=1}^{j-1} P_\zeta(i,k) \le s < \sum_{k=1}^{j} P_\zeta(i,k) \Bigr\}
\]
with the convention that '``$\sum_{k=1}^0[\cdot]_k =0$''.  
We then take, for any $t$,
\[
G_t^\bullet = \clG(0, N^\bullet_t) 
\]
A third i.i.d.\ sequence $\bfmN$ is obtained by sampling:
\begin{equation}
N_{t+1}= 
\begin{cases}
N^\bullet_{t+1} & \text{if \ } \piload_t=\piload_t^\bullet 
\\
N^\circ_{t+1} & \text{else \ }   
\end{cases}
\label{e:Ncouple}
\end{equation}   
We then take $
G_t = \clG(\zeta_{t-1}, N_t)$. 

Based on these definitions, the evolution equation \eqref{e:piloadJoint} holds for some $F$; we now focus analysis on a stationary solution defined for all $t\in\ZZ$.

Choose $T_0\ge 1$,   $\delta_0>0$,      $\epsy_0>0$, so that
\[
\begin{aligned}
\Prob\{ X_{t+1} = X_{t+1}^\bullet \mid X_{t} = X_{t}^\bullet \} 
		& \ge 1- \delta_0 \epsy
\\
\Prob\{ X_{t+T_0} = X_{t+T_0}^\bullet \mid X_{t} \neq X_{t}^\bullet \}   
		&\ge \delta_0,\qquad\qquad t\in\ZZ,\ \epsy\le \epsy_0
\end{aligned}
\]
This is possible by the construction of the joint evolution equations, and the assumption that $\bfmX^\bullet$ is irreducible and aperiodic.  The first bound may be extended to obtain,
\[
\Prob\{ X_{t+T} = X_{t+T}^\bullet \mid X_{t} = X_{t}^\bullet \}     \ge (1- \delta_0 \epsy)^{T} \ge 1-T\delta_0 \epsy,\quad T\ge 1.
\]

We then have by stationarity,
\[
\begin{aligned}
\Prob\{ X_0 = X_0^\bullet  \} =
\Prob\{ X_{T_0} = X_{T_0}^\bullet  \} 
&=\Prob\{ X_{T_0} = X_{T_0}^\bullet \mid X_{0} = X_{0}^\bullet  \} 
\Prob\{   X_{0} = X_{0}^\bullet  \} 
\\
  &\ \ + \Prob\{ X_{T_0} = X_{T_0}^\bullet \mid X_{0} \neq X_{0}^\bullet  \} 
\Prob\{   X_{0} \neq X_{0}^\bullet  \} 
\end{aligned}
\]
Now, substitute the prior bounds,  
giving
\[
\begin{aligned}
\Prob\{ X_0 = X_0^\bullet  \}  
&\ge (1-T_0\delta_0 \epsy) 
\Prob\{   X_{0} = X_{0}^\bullet  \} 
\\
  &\ \ + \delta_0\Prob\{   X_{0} \neq X_{0}^\bullet  \} 
\end{aligned}
\]
Substituting $\Prob\{   X_{0} \neq X_{0}^\bullet  \}  = 1 - 
\Prob\{   X_{0} = X_{0}^\bullet  \} $ and
rearranging terms gives,
\[
\Prob\{ X_0 = X_0^\bullet  \}   \ge \frac{1}{1+T_0 \epsy}\ge 1-T_0  \epsy 
\]
which completes the proof of \eqref{e:piloadpiload}.

The approximation \eqref{e:piloadzeta} follows from \eqref{e:piloadpiload} and independence of $\bfzeta$ and $\bfpiload^\bullet$:
\[
\Expect[\piload_t  \zeta_t  ] 
= \Expect[\zeta_t  \piload_t^\bullet ] + \Expect[\zeta_t ( \piload_t - \piload_t^\bullet) ]
= \Expect[\zeta_t] \Expect[  \piload_t^\bullet ] +O(\epsy^2)  
=  O(\epsy^2) 
\]
The remaining bounds in \eqref{e:BVDeltaBdds}  follow directly from \eqref{e:piloadpiload} and the smoothnes assumptions on $P_\zeta$.
\qed

\section{Martingale difference sequence}
\label{t:proof_CovDecomp}

Here we give the  proof of \Proposition{t:CovDecomp}.
Define the $\sigma$-algebra $\clF_t = \sigma \{\zeta_{-\infty}^{\infty}, N_{-\infty}^t \}$.  The random vector $\Delta_t$ is $\clF_t $-measurable for each $t$, and it follows from \eqref{e:ExpectG=P} that 
\[
	\Expect [\Delta_{t+1}|\clF_t] = \Expect [\piload_t[\clG(\zeta_t, N_{t+1}) - P_{\zeta_t}] | \clF_t]  = 0
\]
This proves that ${\Delta_t}$ is a martingale difference sequence.  Moreover,
using the smoothing property of conditional expectation, for any $t$ and $\tau$,
\[
\begin{aligned}
	\Expect [\Delta_{t+1} \zeta_\tau] &=  \Expect [\piload_t[\clG(\zeta_t, N_{t+1}) - P_{\zeta_t}]\zeta_\tau] 
	\\
					       &=  \Expect [ \piload_t \zeta_\tau ~\Expect [ \{\clG(\zeta_t, N_{t+1}) - P_{\zeta_t} \} | \clF_t] ]=0.
\end{aligned}
\]
This establishes $R_{\Delta,\zeta}(k) = 0$  for any $k$, which is 
\eqref{e:RDz0}.

To obtain the covariance of $\bfDelta$, we first compute the second moment of $\piload_tG_{t+1}$.
We use the fact that $\piload_t$ and $G_{t+1}$ have entries zero or one in the following:
\[
\begin{aligned} 
\Expect[ (\piload_tG_{t+1})^\transpose \piload_tG_{t+1}]_{ij}  
&= 
\sum_{k=1}^d  \Expect[ \piload_t(x^k) G_{t+1}(k,i)  G_{t+1}(k,j) ]
\\
&= 
\sum_{k=1}^d  \Expect[ \piload_t(x^k) G_{t+1}(k,i)  ]\ind\{i=j\} 
\end{aligned}
\]
We also have $\Expect[ \piload_t G_{t+1}   ]  =\Expect[ \piload_t  P_{\zeta_t} ]  $.
These two identities imply the formula, 
\[
\Expect[ (\piload_tG_{t+1})^\transpose \piload_tG_{t+1}] = \diag(\Expect[\piload_t P_{\zeta_t}])
		= \diag(\Expect[\piload_{t+1}]).
\]
The result \eqref{e:smoothedSigmaDelta} follows from this second moment formula:
\[
\begin{aligned}
 \Cov ( \piload_t[G_{t+1}-P_{\zeta_t}] ) &= 
\Expect[ (\piload_tG_{t+1})^\transpose \piload_tG_{t+1}] 
-
\Expect[ (\piload_tP_{\zeta_t})^\transpose \piload_t P_{\zeta_t}]
\\
  & =\diag(\Expect[\piload_{t+1}])
  	- \Expect[ (P_{\zeta_t})^\transpose \diag(\piload_t )P_{\zeta_t}].
\end{aligned}
\]
\qed

\section{Approximating the steady-state mean}
\label{s:pi-epsy}

The approximation  \eqref{e:pi-epsy} is obtained here,  starting with the approximate evolution equation that was used to obtain  \eqref{e:piD}:
\[
\begin{aligned}
\piload_{t+1} & = \piload_t P_{\zeta_t} + \Delta_{t+1} \\
 &=\piload_t [ P_0 + \zeta_t\clE +\half \zeta_t^2 \clW] + \Delta_{t+1} + O(\epsy^3)
\end{aligned}
\]
Taking the mean of each side, and using stationarity,
\begin{equation}
\begin{aligned}
\Expect[\piload_t] =
\Expect[\piload_{t+1} ]
 &= \Expect[\piload_t]   P_0 
 \\
  & \ \ + \Expect[\zeta_t\piload_t ]\clE 
  \\
  & \ \  +\half \Expect[\zeta_t^2  \piload_t ]\clW + O(\epsy^3) 
\end{aligned}
\label{e:piepsyFixedPt}
\end{equation}
To approximate $\Expect[\zeta_t^2  \piload_t ]$ we use $\piload_t = \piload_t^\bullet +\tilO(\epsy)$.  This combined with independence of $\bfpiload^\bullet ,\bfzeta$ gives, 
\[
\Expect[\zeta_t^2  \piload_t ] 
= \Expect[\zeta_t^2  \piload_t^\bullet ] + \Expect[\zeta_t^2 ( \piload_t - \piload_t^\bullet) ]
= \Expect[\zeta_t^2] \Expect[  \piload_t^\bullet ] +O(\epsy^3)  
=\sigma^2_{\zeta} \pi_0 +O(\epsy^3) 
\]
Substituting this   into
\eqref{e:piepsyFixedPt} gives the approximate  
fixed-point equation,
\begin{equation}
\tilpi_\epsy
=  \tilpi_\epsy P_0  + \xi  + O(\epsy^3) 
\label{e:pi_epsy_poi}
\end{equation}
where   $\tilpi_\epsy = \pi_\epsy -\pi_0$,  and $\xi$ is defined in \eqref{e:xiUnderA1A2}.

The matrix $I-P_0$ is not invertible since it has eigenvalue at $0$.  To obtain an invertible matrix we note that $\tilpi_\epsy [\One\otimes\pi_0] =0$ for any $\epsy$,  and hence \eqref{e:pi_epsy_poi} is equivalent to the vector equation,
\[
\tilpi_\epsy [I-P_0+\One \otimes \pi_0] = \tilpi_\epsy [I-P_0] = \xi \epsy^2+ O(\epsy^3) 
\]
The desired result \eqref{e:pi-epsy}
 is obtained by inversion:
\[
\tilpi_\epsy = \xi  [I-P_0 +\One \otimes \pi_0] ^{-1}   + O(\epsy^3) 
			= \xi  U_1  + O(\epsy^3) 
\]
 \qed

\section{Cross-covariance with $\bfzeta$}

 Approximations for $R_{D, \zeta} $ and  $R_{B\zeta}$ are relatively simple because $\zeta_t=O(\epsy)$.

\subsection{Cross-covariance between $D$ and $\zeta$}

Using the coupling result we obtain an  approximation for the cross-correlation function,
\begin{equation}
\begin{aligned}
R_{D,\zeta}(t) & = \Expect [D_{t+1}^\transpose \zeta_1]
\\
                        &= \Expect[(B_t^\transpose \zeta_t +\Delta_{t+1})^\transpose \zeta_1] +O(\epsy^3)
\\
		      &= \Expect[B_t^\bullet   \zeta_t \zeta_1] +O(\epsy^3)
\\
                        &= B R_\zeta(t-1) + O(\epsy^3).		             
\end{aligned}
\label{t:RDzeta-proof}
\end{equation}

\subsection{Cross-covariance between $\bfDelta\bfDelta^\transpose$ and $\bfzeta$}
\label{s:e:RD2z}

\spm{2016:  commented out statement regarding $R_{\Delta^2,\zeta}(-t)$.  We don't restrict $t$ to be positive}

Recall the $\sigma$-algebra $\clF_s = \sigma \{\zeta_{-\infty}^{\infty}, N_{-\infty}^s\}$
introduced in \Section{t:proof_CovDecomp}.  Taking $s=-1$ we obtain from the smoothing property of the conditional expectation,
\[
R_{\Delta^2,\zeta}(-t)  =\Expect[\Delta_0^\transpose \Delta_0 \zeta_t]= \Expect[\zeta_t \Expect[\Delta_0^\transpose \Delta_0\mid \clF_{-1} ]  ]\,, \quad t\in\ZZ\, .
\]
The conditional expectation is a matrix that is denoted 
\begin{equation}
\begin{aligned}
\clX &= \Expect [\Delta_0^\transpose \Delta_0| \clF_{-1}] 
= \diag  (\piload_{-1} P_{\zeta_{-1}}) 
			- P_{\zeta_{-1}}^\transpose \Lambda_{-1}^{\piload } P_{\zeta_{-1}}
\end{aligned}  
\label{e:clX}
\end{equation}

\begin{lemma}
\label{t:clXzeta}
For $t\ge 0$ we have 
\[
\begin{aligned}
R_{\Delta^2,\zeta}(-t) &=\Expect [ \zeta_t \clX  ]
\\
&= 
\diag  (R_{\piload,\zeta}(-t-1) P_0) 
	- P_0^\transpose \diag (R_{\piload,\zeta}(-t-1))P_0
	+ R_\zeta(t+1) \Expect [ \clX_\bullet^{(1)} ] +O(\epsy^3) \,,
\end{aligned}
\]
where $\Expect [ \clX_\bullet^{(1)} ] $ is defined below 
\eqref{e:RD2z}.
\end{lemma}
\begin{proof}
A first order Taylor series approximation gives
$
\clX= \clX^{(0)} + \zeta_{-1} \clX^{(1)} + O(\epsy^2)$, where   
\begin{align*}
 \clX^{(0)} 
 	&=   \diag (\piload_{-1} P_0) - P_0^\transpose \diag  (\piload_{-1}) P_0 
 \\
 \clX^{(1)} 
 	&  = \diag (\piload_{-1} \clE) -\left(P_0^\transpose \diag (\piload_{-1}) \clE  
	+ [P_0^\transpose\diag (\piload_{-1}) \clE]^\transpose \right).
\end{align*}
Hence, for $t\ge 0$,
\begin{align*}
\Expect[\zeta_t \clX] & = \Expect [\zeta_t (\clX^{(0)} + \zeta_{-1} \clX^{(1)} + O(\epsy^2)) ] 
\\
					      &= \Expect [ \zeta_t \clX^{(0)} ] + \Expect [ \zeta_t \zeta_{-1} \clX_\bullet^{(1)} ] +O(\epsy^3)
\\
					      &= \Expect [ \zeta_t \clX^{(0)} ]
					       + R_\zeta(t+1)\Expect [ \clX_\bullet^{(1)} ] +O(\epsy^3)
\end{align*}
where, in the second equality we used $\clX^{(1)} = \clX_\bullet^{(1)} +\tilO(\epsy)$ with
\[ 
\clX_\bullet^{(1)}= \diag (\piload^\bullet_{-1} \clE) -\left(P_0^\transpose \diag (\piload^\bullet_{-1}) \clE  + [P_0^\transpose\diag (\piload^\bullet_{-1}) \clE]^\transpose \right)\,,
\]
and also used the fact that $\piload^\bullet_{-1}$ is independent of $\bfzeta$.   We   have by the definitions,
\[
\Expect [ \zeta_t \clX^{(0)} ] = \diag  (R_{\piload,\zeta}(-t-1) P_0) - P_0^\transpose \diag (R_{\piload,\zeta}(-t-1))P_0.  
\]
Substitution into the previous approximation for $\Expect[\zeta_t \clX]$ completes the proof.
\end{proof}

\subsection{Auto-correlation of $\bfmB \bfzeta$}
\label{s:e:R_bz}

Applying the coupling result \eqref{e:BVDeltaBdds} in \Proposition{t:piload-zeta} gives $B^\transpose_t = {B^{\bullet}}^\transpose_t + \tilO(\epsy)$, where 
  $B^\transpose_t = \piload_t \clE$.   Hence,
\begin{equation}
R_{B\zeta}(t) =\Expect[B_t\zeta_tB_0^\transpose \zeta_0]= \Expect[B_t^\bullet \zeta_t (B_{0}^\bullet \zeta_0)^\transpose]  + O(\epsy^3) =  R_{B^\bullet \zeta}(t) + O(\epsy^3) 
\label{e:Proof:e:R_bzA}
\end{equation}
Independence of $\bfzeta$ and $\piload_t^\bullet $ implies a formula  for the simpler auto-correlation:
\begin{equation} 
\begin{aligned}
R_{B^\bullet \zeta}(t) &\eqdef \Expect[B^\bullet_t \zeta_t (B^\bullet_0 \zeta_0)^\transpose]
\\
& =  \Expect[ (\piload_t^\bullet \clE)^\transpose ( \piload_t^\bullet \clE)] \Expect[\zeta_t \zeta_0]
\\
&= \clE^\transpose R_{\piload^\bullet}(t)\clE R_\zeta(t)
\end{aligned}
\label{e:Proof:e:R_bzB}
\end{equation}
A formula for $R_{\piload^\bullet}(t)$ is given next: 
For $t\ge 0$,  
\begin{equation} 
R_{\piload^\bullet}(t) = \Expect [ (\piload_t^\bullet)^\transpose \piload_0^\bullet ] = \Expect [ (\piload_0^\bullet P_0^t)^\transpose \piload_0^\bullet] = (P_0^\transpose)^t \Expect[(\piload_0^\bullet)^\transpose \piload_0^\bullet] = (P_0^\transpose)^t \Pi_0
\label{e:Proof:e:R_bzC}
\end{equation}
where the last equality is from the fact that $\piload_t$ has binary entries and $\Expect[\piload_t] = \pi_0$.

Combining \eqref{e:Proof:e:R_bzA}--\eqref{e:Proof:e:R_bzC}
completes the proof of \eqref{e:R_bz}.
\qed

\subsection{Cross-covariance between $\piload$ and $\bfzeta$}
\label{s:RGzGen}

The derivation of  \eqref{e:RGzGen} is obtained via a recursion, similar to the calculation in \Section{s:pi-epsy}.  It is simplest to work with the row vectors $\nu_k =  \bigl( R_{\piload,\zeta}(k) \bigr)^\transpose = \Expect [\piload_{k} \zeta_0] $.
For any $k\in\ZZ$,
\begin{align*}
	\nu_{k+1}&= \Expect [\piload_{k+1} \zeta_0] \\
	        &= \Expect[(\piload_k P_{\zeta_k} + \Delta_{k+1}) \zeta_0]\\
	        & = \Expect[\piload_k P_{\zeta_k} \zeta_0]\\
	        & = \Expect[\piload_k (P_0 + \clE \zeta_k +O(\epsy^2) ) \zeta_0]\\
	         &= \nu_k P_0 +  \Expect [ \piload_k \clE \zeta_k \zeta_0 ]+O(\epsy^3)
\end{align*}
where in the third equation we used the fact that the sequences $\bfDelta$ and $\bfzeta$ are uncorrelated.
Recalling the definition $B^\transpose = \pi \clE = \Expect[\piload_k^\bullet] \clE$, 
and applying the coupling result $\piload_k = \piload_k^\bullet + \tilO(\epsy)$ gives
\[
\Expect[\piload_k \clE \zeta_k \zeta_0 ] = \Expect[\piload_k^\bullet \clE \zeta_k \zeta_0 ] +O(\epsy^3) = B^\transpose R_\zeta (k)+ O(\epsy^3).
\]
Hence, $\nu_{k+1} =  \nu_k P_0 + B^\transpose R_\zeta (k)+ O(\epsy^3)$;  that is, there is a bounded sequence of row vectors $\{\nudist_k\}$ such that 
\[
\nu_{k+1} = \nu_k P_0 + B^\transpose R_\zeta (k)+ \epsy^3 \nudist_k
\]
It follows from 
\Lemma{t:FP} that $B^\transpose \One=0$.  Moreover, 
since $\piload_{t-n}$ is a pmf, we have for any $k,\ell$,
\begin{equation}
\nu_k\One = \Expect[\zeta_k \piload_{\ell} ]  \One  =\Expect[\zeta_k]=0.
\label{e:RpiloadOne}
\end{equation}
It then follows that $\nudist_k\One =0$ for each $k$.

On iterating, we obtain for any integer $n \ge1$, 
\[
\nu_{k+n} = \nu_k P_0^n + \sum_{i=1}^n [B^\transpose +\epsy^3 \nudist_i]P_0^{i-1} R_{\zeta}(k+n-i) .
\]
Now, substitute $t=k+n$, where $t\in\ZZ$ is a fixed integer,  and $n$ is a large positive integer:
\begin{equation}	
	\nudist_t = R_{\piload,\zeta}(t-n) P_0^n +  \sum_{i=1}^n  [B^\transpose +\epsy^3 \nudist_i] P_0^{i-1} R_{\zeta}(t-i) .
\label{e:R_iter}
\end{equation}
For large $n$ we have $P_0^n = \One \otimes \pi_0 + o_e(1)$,  where $o_e(1)\to 0$ geometrically fast as $n\to\infty$.   Consequently,
\[
\begin{aligned}
 R_{\piload,\zeta}(t-n) P_0^n  &= \Expect[\zeta_0 \piload_{t-n} ]  \One \otimes \pi_0  + o_e(1)
 \\
	\nudist_i P_0^{i-1} & =  \nudist_i\One \otimes \pi_0  + o_e(1)
\end{aligned}
\]
We have seen previously following \eqref{e:RpiloadOne} that $\nudist_i\One=0$. It follows similarly that $ \Expect[\zeta_0 \piload_{t-n} ]  \One  =\Expect[\zeta_0]=0$, from which we conclude that
\[
\begin{aligned}
 R_{\piload,\zeta}(t-n) P_0^n &=   o_e(1)
  \\
  \text{\it and}\quad
  \sum_{i=1}^\infty  \|   \nudist_i P_0^{i-1}   R_{\zeta}(t-i) \|  & <\infty
\end{aligned}
\]
This justifies letting $n\to\infty$ in \eqref{e:R_iter} to obtain,
\[
	\nudist_t =   B^\transpose\sum_{i=1}^\infty  P_0^{i-1} R_{\zeta}(t-i) + O(\epsy^3)\,,
\]
which is equivalent to \eqref{e:RGzGen}.
\qed

Based on this result we now establish \Proposition{t:R_Gz}. 
  It is sufficient to establish the following approximation:
  \spm{2016 -- note big reorganization}
\begin{equation}
 R_{\piload,\zeta}(0) 
= \epsy^2    \sum_{k=1}^{n_z} a_k \rho_k[I-\rho_k P_0^\transpose ]^{-1} B + O(\epsy^3)
\label{e:t:R_Gz}
\end{equation}
The representation \eqref{e:xiUnderA3}
 for $\xi$ then follows immediately from the original definition  \eqref{e:xiUnderA1A2}.
 
Recall that $\nu_k =   \Expect [\piload_{k} \zeta_0] $, $k\in\ZZ$.
Under \textbf{A3} we apply \eqref{e:R_zeta} to conclude that for   $t\le 1$ and $i\ge 1$,
\[
R_{\zeta}(t-i) = \epsy^2 \sum_{k=1}^{n_z}a_k\rho_k^{i-t}\, .
\]
For these values of $t$ and $i$ we have $i-t = i+|t|$,  and hence the approximation \eqref{e:RGzGen} gives,
\[
\nu_t
	=  \epsy^2 \sum_{k=1}^{n_z}  \sum_{i=1}^\infty 
	 B^\transpose  P_0^{i-1} a_k\rho_k^{i+|t|} + O(\epsy^3).
\]
Rearranging terms,  and letting $j=i-1$ gives,
\[
\nu_t
	=  \epsy^2 B^\transpose \sum_{k=1}^{n_z} a_k\rho_k^{1+|t|} 
	 \sum_{j=0}^\infty 
	  P_0^{j} \rho_k^{j} + O(\epsy^3)\, .
\]
On setting $t=0$ and taking transposes, this becomes \eqref{e:t:R_Gz}
\qed

\section{Proof of \Theorem{t:R_D}}
\label{t:proof_R_D}

The representation of $R_D(t)$ as the sum of the terms (\ref{e:R_DBz}--\ref{e:R_DV}), plus an error of order $\epsy^3$,  follows immediately from the arguments  preceding the theorem.  The proof of   \eqref{e:R_Dz} is given in  \eqref{t:RDzeta-proof}.

Next we consider each of the terms (\ref{e:R_DBz}--\ref{e:R_DV}) separately.

The approximation for $R_{B \zeta}(t) $ that appears  in \eqref{e:R_DBz} was given in \Section{s:e:R_bz}.   Consider next the covariance \eqref{e:R_DDelta}.

\notes{no need to restate formula for  $R_D(t)$.}

\subsection{Computation of $\Sigma^\Delta$ in \eqref{e:R_Delta}}

Take the expectation inside the $\diag$ operator in \eqref{e:smoothedSigmaDelta} to obtain
\begin{equation}
 \Sigma^\Delta = \Pi_\epsy -\Expect [P_{\zeta_t}^\transpose \diagpiload_t P_{\zeta_t} ],
 \label{e:psd_Delta}
\end{equation}
where, $\Pi_\epsy = \diag (\pi_\epsy)$, with $\pi_\epsy$ estimated in \Proposition{t:pi_epsy+R_Gz}. 

We approximate $\Expect [P_{\zeta_t}^\transpose \diagpiload_t P_{\zeta_t} ]$ within the expectation using a second order Taylor series expansion.  The random matrix $\diagpiload_t$ has binary entries, so we leave it fixed in this approximation.  

For any scalar $\zeta$ and matrix $\Lambda$ we have the Taylor series approximation,
\begin{equation}
\begin{aligned}
 P_\zeta^\transpose \Lambda P_\zeta  =    P_0^\transpose \Lambda P_0   + \zeta M^{(1)}  + \half \zeta^2 M^{(2)} +O(|\zeta|^3) 
\end{aligned}	
\label{e:PLP}
\end{equation}
where $M^{(i)}$ is the $i$th derivative of $P_\zeta^\transpose \Lambda P_\zeta$ with respect to $\zeta$, evaluated at $\zeta=0$. 
\def\ddzeta{{\textstyle \frac{d}{d\zeta}}}

To obtain $M^{(1)}$,  first differentiate using the product rule:
\[
\begin{aligned}
\ddzeta  P_\zeta^\transpose \Lambda P_\zeta  
&=
  P_\zeta^\transpose \Lambda P_\zeta'+ [ P_\zeta^\transpose \Lambda P_\zeta']^\transpose  
\end{aligned}
\]
Evaluating at  $\zeta=0$ gives,
\[
M^{(1)} =  P_0^\transpose \Lambda \clE+ [ P_0^\transpose \Lambda \clE ]^\transpose  
\]
Similarly, given the  second derivative,
\[
\begin{aligned}
{ \textstyle \frac{d^2}{d\zeta^2}}
 P_\zeta^\transpose \Lambda P_\zeta   
=
(    {P_\zeta'}^\transpose \Lambda P_\zeta' +   P_\zeta^\transpose \Lambda P_\zeta'')+
(    {P_\zeta'}^\transpose \Lambda P_\zeta' +   P_\zeta^\transpose \Lambda P_\zeta'')^\transpose 
\end{aligned}
\]
we obtain
\[
M^{(2)}
=P_0^\transpose \Lambda \clW + 2  {\clE}^\transpose \Lambda \clE +   \clW^\transpose \Lambda P_0   
\]

Substituting $\zeta_t=\zeta$ and $\Lambda_t=\Lambda$ in
\eqref{e:PLP} gives the approximation,  
 \begin{equation}
 \begin{aligned}
 \Expect[   P_{\zeta_t}^\transpose \diagpiload_t P_{\zeta_t} ] 
 &= \Expect [ P_0^\transpose \diagpiload_t P_0 ]
 \\
 &
 \quad + \Expect[ (P_0^\transpose \diagpiload_t \clE  
 		+  {\clE}^\transpose \diagpiload_t P_0 )  \zeta_t ]
 \\
 &\quad + \half \Expect [ ( 
 P_0^\transpose \diagpiload_t \clW + 2  {\clE}^\transpose \diagpiload_t \clE +      {\clW}^\transpose \diagpiload_t P_0 )  \zeta_t^2 ] + O(\epsy^3)
\end{aligned}
\label{e:smoothedSigmaDelta2}
\end{equation}
The first term on the RHS can be computed   using \Proposition{t:pi_epsy+R_Gz}.
The second expectation involves $\Expect[\zeta_t\piload_t]$, which is computed in \Lemma{t:R_Gz}.
For the third term we replace $\piload_t$ by $\piload_t^\bullet + \tilO(\epsy)$, where $\piload_t^\bullet$ has mean $\pi_0$ and is independent of $\bfzeta$.  Combining all of these approximations gives   the following approximation for the second term in \eqref{e:psd_Delta}:
\begin{equation}
 \begin{aligned}
\Expect[  P_{\zeta_t}^\transpose \diagpiload_t P_{\zeta_t} ] 
 &=  P_0^\transpose \Pi_\epsy P_0 \\
 &   \quad + P_0^\transpose \diag ( \Expect [\piload_t \zeta_t]  \clE
 		+  \clE^\transpose \diag ( \Expect [\piload_t \zeta_t]) P_0 \\
&   \quad + \half ( 
 P_0^\transpose \Pi_0 \clW+ 2  {\clE}^\transpose \Pi_0  \clE +  {\clW}^\transpose \Pi_0 P_0 )  \Expect [ \zeta_t^2 ] + O(\epsy^3)
\end{aligned}
\end{equation}
where, $\Pi_0 = \diag (\pi_0)$ and $\Pi_\epsy = \diag (\pi_\epsy)$.  
This gives \eqref{e:R_Delta} since $R_{\piload,\zeta}(0) = \Expect [\piload_t^\transpose \zeta_t]$.
\qed

\subsection{Computation of $ R_{B\zeta,\Delta}(t)$ in \eqref{e:R_bzD}}
\label{t:proof_R_bzD}
This is the most complex part of the proof.

We consider three cases separately:  For  $t< 0$ we have demonstrated that $ R_{B\zeta,\Delta}(t)=0 $.  The second case is $t=0$.   The approximation for $ R_{B\zeta,\Delta}(0)$ is used as an initial condition in a recursive approximation for  $ R_{B\zeta,\Delta}(t)$,  $t\ge 1$.

An approximation for $R_{B\zeta, \Delta}(0) $ is obtained from  \Lemma{t:clXzeta}.
Using  $\Expect[ \Delta_0 | \clF_{-1}]=0$, we obtain 
\begin{equation}
\begin{aligned}
	R_{B\zeta, \Delta}(0) &= \Expect[(\piload_0 \clE \zeta_0)^\transpose \Delta_0] 
			\\
	&= \clE^\transpose \Expect[\zeta_0 (\piload_{-1}P_{\zeta_{-1}} +\Delta_0)^\transpose  \Delta_0] 
					\\
	&= \clE^\transpose \Expect[\zeta_0 (\piload_{-1}P_{\zeta_{-1}} )^\transpose \Delta_0]  +\clE^\transpose R_{\Delta^2,\zeta}(0)
					\\
	&= \clE^\transpose \Expect[ \zeta_0 (\piload_{-1}P_{\zeta_{-1}} )^\transpose \Expect[ \Delta_0 | \clF_{-1}] ]  
		+\clE^\transpose R_{\Delta^2,\zeta}(0)
 \\
	&=  \clE^\transpose R_{\Delta^2,\zeta}(0)
\end{aligned}
\label{e:R_BzD0}
\end{equation}
\Lemma{t:clXzeta} provides an  approximation for $R_{\Delta^2,\zeta}(0)$.  Substituting this
into \eqref{e:R_BzD0} gives the approximation for $R_{B\zeta, \Delta}(0)$   shown in \eqref{e:R_bzD}.
\qed


In the remainder of this subsection we consider  $t\ge 1$.
Iterating \eqref{e:piG} gives,
\[
B_t^\transpose  =\piload_t \clE =   \piload_0 G_{1}  G_{2}\cdots G_t  \clE 
\]
Each term is conditionally independent of $\Delta_0$, given $\bfzeta$, except for $ \piload_0 =   \piload_{-1} G_{0} = \piload_{-1}  P_{\zeta_{-1}} +\Delta_0$.
Using the fact that $ \piload_{-1}  P_{\zeta_{-1}} $ is also conditionally independent of $\Delta_0$, we obtain
\begin{align}
R_{B\zeta,\Delta}(t) 
&=
\Expect[ \zeta_t  \clE^\transpose P_{\zeta_{t-1}}^\transpose  P_{\zeta_{t-2}}^\transpose \cdots  P_{\zeta_0}^\transpose [   \piload_{-1}  P_{\zeta_{-1}} + \Delta_0]^\transpose \Delta_0] \nonumber
\\
&= \clE^\transpose \Expect [\zeta_t P_{\zeta_{t-1}}^\transpose \cdots P_{\zeta_0}^\transpose \Expect [\Delta_0^\transpose \Delta_0| \clF_{-1}] ] \nonumber
\\
&= \clE^\transpose \Expect [\zeta_t P_{\zeta_{t-1}}^\transpose \cdots P_{\zeta_0}^\transpose \clX ] \label{e:R_bzD_A}
\end{align}
where $\clX$ was introduced in \eqref{e:clX}.

Denote $A_i = P_{\zeta_i}^\transpose$,
 $A = P_0^\transpose$, and
the matrix product $\Omega_t=A_t  A_{t-1} \cdots  A_0 $, for   $t\ge 0$. 
We set $\Omega_t=I$ for $t<0$. 
The product  is approximated in the following:    
\begin{lemma}
\label{t:Qt}
For $t\ge 0$,
\begin{equation}
\Omega_t = A^{t+1}  + \sum_{i=0} ^{t} A^{t-i} \clE^\transpose  A^i  \zeta_i + O(\epsy^2).
\label{e:Qt}
\end{equation}
\end{lemma}

\spm{2016:  added comment in response to Yue's question}

The proof of \eqref{e:Qt}
 is postponed to the end of this subsection.

Once we have established this lemma,
we then have the complete cross-correlation:

\begin{proof}[Proof of approximation \eqref{e:R_bzD} for $R_{B\zeta, \Delta}(t)$, $t\ge 1$.]
We return to
\eqref{e:R_bzD_A}, which can be expressed
\[
 R_{B\zeta, \Delta}(t) = \clE^\transpose \Expect [\zeta_t   \Omega_{t-1}  \clX ],\qquad t\ge 1. 
\]
Substituting the bound in \Lemma{t:Qt},
\begin{align*}
R_{B\zeta, \Delta}(t)   
	&= \clE^\transpose  \Expect[ \zeta_t A^t \clX] 
		+   \clE^\transpose
		 \Expect \Bigl[ \zeta_t  \sum_{i=0} ^{t-1} A^{t-1-i} \clE^\transpose  A^i \zeta_i  \clX \Bigr] 
		 	+ O(\epsy^3) 
	 \\
	 &=  \clE^\transpose A^t  \Expect [ \zeta_t \clX  ] 
	 + 
	 \clE^\transpose \Expect \Bigl[ \zeta_t  \sum_{i=0} ^{t-1} A^{t-1-i} \clE^\transpose  A^i \zeta_i  \left( \clX^\bullet  +\tilO(\epsy) \right) \Bigr]+ O(\epsy^3) 
	 \\
	  &=  \clE^\transpose A^t \Expect [ \zeta_t \clX  ] 
	 	+  \clE^\transpose \sum_{i=0} ^{t-1} A^{t-1-i} \clE^\transpose A^i R_\zeta (t-i) \Sigma^{\Delta^\bullet} + O(\epsy^3)
\end{align*}
This establishes  \eqref{e:R_bzD} since  $\Expect [ \zeta_t \clX  ] = R_{\Delta^2,\zeta}(-t)$.
\end{proof}

The proof of \Lemma{t:Qt} uses a Taylor series approximation for $\Omega_t$:
\[
\Omega_{t} =  A_{t} \Omega_{t-1} = A\Omega_{t-1} + W_t\, ,
\]
where $W_t = E_t+F_t$;
$E_t = \clE^\transpose \zeta_t \Omega_{t-1}$ is the first order term in the Taylor series approximation, and $F_t$ is the approximation error whose norm is bounded by $O(\epsy^2)$.   

The following result provides a uniform bound on $A^{t-i}F_i$ for each $t$ and $i$,  and also $
A^{t-i} \clE^\transpose $  (which appears in \eqref{e:Qt}).

\begin{lemma}
\label{t:FP}
The identity $\clE \One =\Zero$ holds.  Consequently, there exists $0<\kappa<\infty$ and $0<\varrho<1$ such that 
\[
\begin{aligned}
\|A^n \clE^\transpose \| = \| \clE P_0^n \|  &\le \kappa \varrho^n  
 \\
 \|  A^n F_i  \|  & \le \kappa \varrho^n  \epsy^2
\end{aligned}
\] 
\end{lemma}

\begin{proof} 
Since 	$P_\zeta \One  = \One $ for all $\zeta$, we have 
\[
\clE \One  = \frac{d}{d\zeta} P_\zeta \Big|_{\zeta=0} \One  
	= \frac{d}{d\zeta} P_\zeta \One  \Big|_{\zeta=0} = \frac{d}{d\zeta} \One  = \Zero .  	
\]
Next we apply the ergodic limit \eqref{e:ergo}, recalling that $P_0^n \to \One  \otimes \pi_0$ geometrically fast as $n \to \infty$:  there exists $0<\kappa_0<\infty$ and $0<\varrho<1$ such that
\begin{equation}
\label{e:conv_pi}
\| e_n\|  \le \kappa_0 \varrho^n, \quad \text{\it with \ } e_n = P_0^n - \One  \otimes \pi_0.
\end{equation}
Consequently,  $\clE P_0^n  = \clE P_0^n  \One  \otimes \pi_0 + \clE e_n = \clE e_n $,
which implies the desired bound $\| \clE e_n \| \le \kappa \varrho^n$,  with $\kappa = \kappa_0 \| \clE\|$.

The proof of the second bound is similar: For each $i$ we have
\begin{align*} 
	&\One ^\transpose W_i = \One ^\transpose (A_{i}-A) \Omega_{i-1} = \Zero ^\transpose \\
	&\One ^\transpose E_i = \One ^\transpose \clE^\transpose \Omega_{i-1}\zeta_i  = (\clE \One )^\transpose \Omega_{i-1}\zeta_i  =\Zero ^\transpose 
\end{align*}
We then have $
	\One ^\transpose F_i = \One ^\transpose (W_i - E_i) =  \Zero ^\transpose$,
from which we obtain as before,
\[
F_i^\transpose P_\zeta^n = F_i^\transpose (\One  \otimes \pi_0 + e_n) = F_i^\transpose  e_n 
\]
Applying \eqref{e:conv_pi}, we arrive at the desired bound:
\[
 \|  A^n F_i  \|  =
\|  F_i^\transpose P_0^n \|   = \|  F_i^\transpose e^n \|  \le  \|  F_i \| \, \| e^n \|   \le \epsy^2 \kappa \varrho^n  ,
\]
where $ \epsy^2 \kappa$ is equal to $\kappa_0$ times a worst-case bound on $ \|  F_i \| $.
\notes{this was said above:
the $O(\epsy^2)$ bound exists since $|\zeta_t^1|\le 1$ a.s.,  and hence  $|\zeta_t|\le \epsy$.}
\end{proof}

\begin{proof}[Proof of \Lemma{t:Qt}.]
Applying \Lemma{t:FP},  
\begin{align*} 
	\sum_{i=1}^{t}\| A^{t-i} F_i\|  &= \sum_{i=1}^{t}\|  F_i^\transpose P^{t-i}_0 \|    \le  \epsy^2\frac{\kappa}{1-\varrho}  = O(\epsy^2)
\end{align*}
This implies the following approximation for $\Omega_t$:
\begin{align*}
\Omega_{t}    = A_{t} \Omega_{t-1} 
              &= (A + \clE^\transpose \zeta_{t}) \Omega_{t-1} + F_t
              		 \\
              &= A \Omega_{t-1} + \clE^\transpose  \Omega_{t-1}\zeta_t  + F_t
              		\\
              &= A^{t} \Omega_0 + \sum_{i=1} ^{t} A^{t-i} \clE^\transpose \Omega_{i-1}\zeta_i  
              	+ \sum_{i=1} ^{t} A^{t-i} F_i\\
              &= A^t(A + \clE^\transpose \zeta_0 + O(\epsy^2))   \\
              & \qquad + \sum_{i=1} ^{t} A^{t-i} \clE^\transpose \Omega_{i-1}\zeta_i  + O(\epsy^2) \\
              &= A^{t+1}  + \sum_{i=0} ^{t} A^{t-i} \clE^\transpose \Omega_{i-1}\zeta_i  + O(\epsy^2)
\end{align*}
In particular, this shows that $\Omega_t= A^{t+1} +O(\epsy)$.  Moreover, \Lemma{t:FP} gives the geometric bound   $\| A^{t-i} \clE^\transpose\| \le \kappa \varrho^{t-i}$.   This justifies substitution    $\Omega_{i-1}\zeta_i  =  A^i \zeta_i +O(\epsy^2)$ to obtain the desired result \eqref{e:Qt}. 
\end{proof}

\subsection{Approximation of $R_{V\zeta^2, \Delta}(t)$ in \eqref{e:R_VzD}}

Recall that $V_t^\transpose = \half  \piload_t \clW$,  and denote  $R_{\piload^\bullet,\Delta^\bullet}(t) = \Expect[ (\piload_t^{\bullet})^\transpose   \Delta_0^\bullet ] $.

Applying the coupling result \Proposition{t:piload-zeta},  the cross-correlation is approximated as follows:
\[
\begin{aligned}
R_{V\zeta^2, \Delta}(t)& = \Expect[\half (\piload_t \clW)^\transpose \zeta_t^2 \Delta_0]
\\
	&   =\half \clW^\transpose  \Expect[\zeta_t^2]  \Expect[ (\piload_t^{\bullet})^\transpose   \Delta_0^\bullet ] + O(\epsy^3)
\\
	&   =\half   \epsy^2 \sigma_{\zeta^1}^2  \clW^\transpose   R_{\piload^\bullet,\Delta^\bullet}(t)+ O(\epsy^3)
\end{aligned}
\]
We have $R_{\piload^\bullet,\Delta^\bullet}(t) =0$ for $t<0$, and thus $R_{V\zeta^2, \Delta}(t) = O(\epsy^3)$ for $t<0$.		

We  also have $R_{\piload^\bullet,\Delta^\bullet}(0) = \Expect[(\piload_{0}^\bullet  )^\transpose \Delta_0^\bullet] =  \Sigma^{\Delta^\bullet}$,  where $\Sigma^{\Delta^\bullet} =  \Pi_0 - P_{0}^\transpose \Pi_0 P_{0}$, with $\Pi_0 = \Pi_0$.   

For  $t \ge 1$,  
\spm{2016:  At Yue's request}
\[
\begin{aligned}
R_{\piload^\bullet,\Delta^\bullet}(t) 
&= \Expect[(\piload_{0}^\bullet P_{\zeta_0} \dots P_{\zeta_{t-1}})^\transpose \Delta_0^\bullet]
\\
& = \Expect[((\piload_{-1} ^\bullet P_{\zeta_{-1}} + \Delta_0^\bullet) P_{\zeta_0} \dots P_{\zeta_{t-1}})^\transpose \Delta_0^\bullet]
\\
& =\Expect[P_{\zeta_{t-1}}^\transpose \dots P_{\zeta_0}^\transpose (\Delta_0^\bullet)^\transpose \Delta_0^\bullet]
\\
& =\Expect[( P_0^\transpose)^t (\Delta_0^\bullet)^\transpose \Delta_0^\bullet] 
 =(P_0^\transpose)^t \Sigma^{\Delta^\bullet}\, .
\end{aligned}
\]
Substituting
$R_{\piload^\bullet,\Delta^\bullet}(t) = (P_0^\transpose)^t \Sigma^{\Delta^\bullet}
$ into the previous expression for $R_{V\zeta^2, \Delta}(t)$ gives   \eqref{e:R_VzD} for $t\ge0$.

\subsection{Proof of \Proposition{t:psdDecomp}}
\label{s:psdDecomp}
 
The proof begins with  the uniform bound,
\[
\|\psd(\theta) - 
\hapsd(\theta) \| \le  \sum_{t=-\infty}^{\infty} \| \Sigmatot(t) - \haSigmatot(t) \| ,  \quad \theta\in\Re
\]
where $\|\varble\|$ is any matrix norm.  The right hand side is finite under Assumptions A1 and A2.  It remains to obtain an estimate that is $O(\epsy^{2+\varrho})$.  We establish a slightly stronger bound,
\begin{equation}
 \sum_{t=0}^{\infty} \| \Sigmatot(t) - \haSigmatot(t) \| = O(\epsy^2\log(1/\epsy)).
\label{e:epsylogepsybdd}
\end{equation}

Under the assumption that $R_\zeta(t)\to 0$ geometrically fast, it follows that the same is true for $\haSigmatot(t)$ and $\Sigmatot(t)$:  for some $b_0<\infty$ and $\delta>0$,
\[
\| \Sigmatot(t)\| +
\| \haSigmatot(t)\| \le \exp(b_0 - \delta |t|),\qquad t\in \ZZ.
\]
Moroever,      $\mu = \mu_0+  O(\epsy^2)$,  where the first $d$ components of $\mu_0$ coincide with $\pi_0$, and the remaining are zero.  It follows that $\Sigmatot(t) = R(t) -\mu_0\mu_0^\transpose+O(\epsy^4)$.  
Consequently,  
\Theorem{t:R_D} implies that for some $b_1<\infty$,
\[
\| \Sigmatot(t) - \haSigmatot(t)\| \le \exp(b_1) \epsy^3,\qquad t\in \ZZ.
\]
To establish \eqref{e:epsylogepsybdd} we decompose the sum into two parts.  Denote
\[
N(\epsy) = \min\{ t\ge 0 : \exp(b_0 - \delta t) \le \exp(b_1) \epsy^3\}  .
\]
This implicit definition can be solved to give $N(\epsy) = [b_0 + 3b_1\log(1/\epsy)]\delta^{-1}$.

From this we obtain,
\[
\begin{aligned}
 \sum_{t=0}^{\infty} \| \Sigmatot(t) - \haSigmatot(t) \| 
  &\le   \exp(b_1) \epsy^3 N(\epsy) +  \sum_{t> N(\epsy)} \exp(b_0 - \delta t) 
  \\
  &\le   \exp(b_1) \epsy^3 N(\epsy) +   \exp(b_0 - \delta N(\epsy) )  \frac{1}{1-  \exp( -\delta ) }
  \\
  &\le   \exp(b_1) \epsy^3  N(\epsy)  +    \exp(b_1) \epsy^3  \frac{1}{1-  \exp( -\delta ) }
\end{aligned}
\]
This together with the formula for $N(\epsy)$ immediately gives the bound \eqref{e:epsylogepsybdd}.
\qed

\section{Proof of \Proposition{t:haD}}
\label{s:haD}

We first recall the representation of relative entropy as the convex dual of the log-moment generating function:  For any probability measure   $\probB$ on $\clE$ and measurable function $f\colon\clE\to\Re$, denote the log-moment generating function,
\[
\LpB(f) = \log ( \probB(e^f))
\]
For any other probability measure   $\probA$ on $(\clE,\clB)$ we have,
\[
\KL(\probA\|\probB) = \sup\{ \probA(f) -\LpB(f) \}
\]
where the supremum is over all measurable functions $f$ for which $\probA(f) $ is defined \cite[Theorem 3.1.2]{demzei98a}.

Provided the relative entropy is finite, the supremum is achieved uniquely with the log-likelihood function $f^*=\log(d\probA/d\probB)$.  The error bound (ii) in the proposition is vacuous unless $\| f^*\|_\infty<\infty$.    Consequently, to establish the error bound we can restrict to functions $f$ for which $\| f\|_\infty<\infty$.  

We apply the second order Taylor-series expansion,
\[
\LpB(f) = \log ( 1+  \probB(f) + \half \probB(f^2)) =  \probB(f) + \half \probB(f^2) + O(\| f\|_\infty^3).
\]
A quadratic approximation to relative entropy is obtained on dropping the error term.  To complete the proof we establish the following alternate expression for \eqref{e:haD}:
\[
\haKL(\probA\|\probB) \eqdef \sup\{ \probA(f) -\probB(f)-  \half\probB(f^2) \}
\]
where the supremum is over all functions $f$ whose mean is defined with respect to both $\probA$ and $\probB$.
Without loss of generality we may assume that the maximum is over all functions $f$ for which $\probB(f)=0$.  
It is not difficult to show that the maximizing function is $\haf^* =  e^{f^*}-1 = d\probA/d\probB -1$ whenever $\haKL(\probA\|\probB) $ is finite.  This establishes (i). 



We can scale by a constant  $\theta$ to obtain
\[
\haKL(\probA\|\probB) = \max_f \max_\theta\{ \theta \probA(f) -  \half \theta^2\probB(f^2) \}
\]
where the first maximum is over measurable functions $f$ satisfying $\probB(f)=0$ and  $ \probA(|f|) +\probB(f^2)<\infty$.
The optimizing value    is $\theta^*= \probA(f) /\probB(f^2)$.  Substitution leads to the  formula \eqref{e:haD}.
\qed



\end{document}